\documentclass[12pt, draftclsnofoot, onecolumn]{IEEEtran}
\usepackage{epsfig,graphics,amssymb}
\usepackage{amscd}
\usepackage{amsmath,bbm}
\usepackage{mathtools}
\usepackage{enumerate}
\usepackage{cite}
\usepackage[font={small}]{caption}
\usepackage{verbatim}
\usepackage{color}
\newtheorem{theorem}{Theorem}
\newtheorem{lemma}{Lemma}

\newtheorem{definition}{Definition}

\usepackage[process=auto]{pstool}
\usepackage{tikz}
\newcommand{\T}{{\rm T}}

\newcommand{\convp}{{\xrightarrow{\text{p}}}}

\newcommand\convleq{\stackrel{\mathclap{\normalfont\mbox{$p$}}}{\leq}}

\newcommand{\matr}[1]{\mathlette{\boldmath}{#1}}

\def\argmax{\mathop{\rm argmax}}

\def\expect{\mathop{\mbox{$\mathsf{E}$}}}

\newcommand{\e}{{\rm e}}

\newcommand{\define}{\stackrel{\triangle}{=}}

\newcommand{\dif}{{\rm d}}

\begin{document}

\title{On User Pairing in NOMA Uplink}
\author{\IEEEauthorblockN{Mohammad A. Sedaghat, and Ralf R. M\"{u}ller, \emph{Senior Member, IEEE}}
\thanks{Mohammad A. Sedaghat and Ralf R. M\"uller are with Friedrich-Alexander Universit\"at Erlangen-N\"urnberg (e-mails: mohammad.sedaghat@fau.de, ralf.r.mueller@fau.de).}
}

\maketitle
\begin{abstract}
User pairing in Non-Orthogonal Multiple-Access (NOMA) uplink based on channel state information is investigated considering some predefined power allocation schemes. The base station divides the set of users into disjunct pairs and assigns the available resources to these pairs. The combinatorial problem of user pairing to achieve the maximum sum rate is analyzed in the large system limit for various scenarios, and some optimum and sub-optimum algorithms with a polynomial-time complexity are proposed. In the first scenario, $2M$ users and the base station have a single-antenna and communicate over $M$ subcarriers. The performance of optimum pairing is derived for $M\rightarrow \infty$ and shown to be superior to random pairing and orthogonal multiple access techniques. In the second setting, a novel NOMA scheme for a multi-antenna base station and single carrier communication is proposed. In this case, the users need not be aware of the pairing strategy. Furthermore, the proposed NOMA scheme is generalized to multi-antenna users. It is shown that random and optimum user pairing perform similarly in the large system limit, but optimum pairing is significantly better in finite dimensions. It is shown that the proposed NOMA scheme outperforms a previously proposed NOMA scheme with signal alignment.
\end{abstract}

\section{introduction}
Non-Orthogonal Multiple-Access (NOMA) has been introduced recently as one of the key technologies for the next generation of wireless networks referred to as \textit{5G} \cite{saito2013non,ding2014performance,ding2016impact}. Compared to Orthogonal Multiple-Access (OMA) techniques, NOMA allows a larger number of users to be connected to a wireless network simultaneously \cite{ding2017application}. Furthermore, it has been shown that NOMA achieves a larger total throughput in both uplink and downlink of a wireless network compared to OMA techniques \cite{chen2016mathematical,tabassum2016non}. Compared to conventional opportunistic schemes, a higher fairness is achieved by NOMA \cite{timotheou2015fairness,ding2017application}. The low latency property of NOMA together with realizing a tradeoff between total throughput and fairness makes NOMA very attractive for 5G and a serious competitor for well-known Orthogonal Frquency Division Multiple-Access (OFDMA) \cite{wei2016survey,saito2013non}.

NOMA can be done in power domain by allowing more than one users to send over a resource block \cite{benjebbour2013system} and also in code domain by assigning different sparse codes to users, known as Sparse Code Multiple-Access (SCMA) \cite{nikopour2013sparse, nikopour2014scma}. This paper is about power domain NOMA, therefore in the remaining parts of this paper, by NOMA we refer to power domain NOMA.


The idea of NOMA is based on the fact that users normally have different channel qualities and rate demands. Information theory shows that orthogonal signaling techniques are not optimum in such cases \cite{tse:05,caire2007hard,higuchi2015non}. In NOMA, base stations share each resource among some users with different channel qualities, and the receiver in a NOMA system uses Successive Interference Cancellation (SIC) to detect the transmitted signals. It is known from information theory that, for a given set of rates and users with different channel qualities, OMA methods are suboptimum with respect to required total transmit power on multiple-access channels \cite{tse:05,caire2007hard}. Therefore, NOMA methods can be more power efficient than OMA methods when users with different channel qualities are grouped. Different channel qualities of the users typically result from fading and the near-far effect in cellular networks. 

NOMA with randomly deployed mobile stations in downlink is discussed in \cite{ding2014performance}, and it is shown that NOMA achieves higher ergodic sum rates compared to OMA schemes. Impact of user pairing in NOMA downlink is investigated in \cite{ding2016impact}. Cooperative NOMA in downlink is introduced in \cite{ding2015cooperative} in which after the direct transmission phase, users cooperate via short range communication channels. Fairness in NOMA is investigated in \cite{timotheou2015fairness} by studying some power allocation techniques based on instantaneous and statistical channel state information. Joint user pairing and power allocation in a multi-carrier downlink channel with single-antenna users and base stations is addressed numerically in \cite{sun2016optimal} using some optimization techniques. The authors in \cite{chen2016mathematical} provide a mathematical proof for the superiority of NOMA in downlink channel compared to OMA schemes from an optimization point of view. Furthermore, in \cite{he2016fast}, an efficient user pairing and power allocation scheme with low complexity is proposed for NOMA downlink.

In NOMA downlink, SIC should be implemented at user terminals which needs high processing power especially when the number of active users at each resource block is large. One should also consider the loss of practical SIC for large number of users \cite{mueller:99}. To this end, grouping only two users at each resource block in NOMA downlink has become more popular. In uplink, SIC is at base stations, therefore users need not be aware of the modulation and coding schemes employed by the other users. Furthermore, base stations have enough processing power to perform SIC. NOMA can also allow users to transmit in uplink in a grant-free manner which reduces latency significantly.

NOMA uplink using a sub-optimal multi-user detection technique is investigated in \cite{al2014uplink}, and an iterative detection and decoding scheme is proposed in \cite{al2015receiver} which improves the system performance significantly. Furthermore, an iterative method for joint subcarrier and power allocation is proposed in \cite{al2015receiver}. 

Along with single-antenna applications, NOMA has been also proposed for multi-antenna wireless networks. Ergodic capacity of Multiple-Input Multiple-Output (MIMO) NOMA is discussed in \cite{sun2015ergodic} by studying optimum power allocation. The authors in \cite{ding2016application} propose a new MIMO NOMA scheme by designing the precoding and detection matrices. In \cite{ding2016application}, the number of antennas at each user terminal should be larger than
the number of antennas at base stations. A more general MIMO NOMA is introduced in \cite{ding2016generalTWCOM} based on signal alignment for a multi-user MIMO system in which both the users and the base station are equipped with antenna arrays. It has been shown that the number of supported users in a single frequency band can be up to twice the number of antennas at the base station under some certain conditions \cite{ding2016generalTWCOM}. The proposed method in \cite{ding2016generalTWCOM} is applicable to both downlink and uplink. A NOMA scheme for massive MIMO with limited feedback is proposed in \cite{ding2016design} in which a massive MIMO NOMA channel is decomposed into multiple separated single-input single-output NOMA channels.


Although, user pairing and power allocation for NOMA uplink have been studied previously, to the best of our knowledge there is no analytical performance evaluation for optimum user pairing in NOMA uplink. In this paper, we answer the question how much performance improves if the users in NOMA uplink are paired optimally. Furthermore, we show that optimal user pairing in uplink can be implemented using some algorithms with polynomial time complexity. To this end, we consider the general problem of user pairing in NOMA uplink in a cellular network with some fixed power allocation schemes. We consider the investigation of optimum joint power allocation and user pairing as a future work.
At each cell, a central base station assigns the available resources including some subcarriers and spatial dimensions to some users to achieve the best performance in the sense of total rate of the users. The inter-cell interference is neglected and a single isolated cell is considered. The user terminals and the base station in the considered cell may all have multi-antenna arrays. The base station first divides the users into some groups and lets the users at each group transmit over one of the subcarriers. Then, the signals at each subcarrier are detected by dividing the users, transmitting at each subcarrier, into some pairs and applying SIC to each pair.

 Solving such a problem in general is very complicated. Thus, in this paper we consider various scenarios. The contributions of the paper for these scenarios are summarized as follows.
\begin{itemize}
\item {\bf Single antenna multiple subcarriers (SAMS):} In the SAMS setting, the users and the base station are assumed to have a single-antenna. The users are divided into disjunct pairs and every subcarrier is assigned to a pair. A sub-optimum pairing algorithm with polynomial time complexity is proposed and shown to have almost the same performance as optimum method. Furthermore, a large system performance analysis of optimum pairing for the case that the users have independent and identically distributed (iid) Gaussian small scale fading coefficients at different subcarriers, is presented. It is shown that for $2M$ users and $M$ subcarriers, the total rate of the system divided by $M\log_2\log M$ converges to $1$ in probability as $M\rightarrow \infty$. Furthermore, for frequency-flat channels and large $M$, the performance of optimum pairing is derived and compared against the performance of random pairing and time division multiple-access (TDMA) at each subcarrier.

\item {\bf Multiple base station antennas single subcarrier (MBASS):}
\begin{itemize}
\item The MBASS-SAU setting is the case of a multi-antenna base station and single-antenna users (SAU). We consider one single subcarrier and investigate the performance of optimum spatial dimension assignment. The base station divides the users into pairs and detects every pair using SIC. It is shown that the optimum pairing method has significant performance gain compared to random pairing in finite dimensions. We show that optimum assignment in this case can be implemented in polynomial time using the Hungarian algorithm with a modified cost matrix. Furthermore, a large system analysis of optimum pairing method is presented and it is shown that when the number of users and the number of antennas at the base station are large, the performances of both  optimum and random pairing converge to the same limit.
\item Motivated by the work in \cite{ding2016application,ding2016generalTWCOM}, the case of multi-antenna users (MAU) is considered as the MBASS-MAU setting. A new polynomial time NOMA scheme is proposed in which the users need not know the user pairing strategy and the base station pairs the users based on channel state information. It is shown that, unlike to NOMA with signal alignment \cite{ding2016generalTWCOM}, there are no requirements on the ranks of the channels and the number of antennas at the user terminals. Furthermore, the users only need to know their own channel. We show that the proposed method has better performance than NOMA with signal alignment.
\end{itemize}
\end{itemize}

We use bold lowercase letters for vectors and bold uppercase letters for matrices. Conjugate transpose of a matrix is denoted by $\cdot^\dagger$ and the transpose itself is shown by $\cdot^{\sf T}$.  The identity matrix is shown by $\matr{I}$. Moreover, the real and complex sets are shown by $\mathbb{R}$ and $\mathbb{C}$, respectively. The determinant is denoted by $\det(\cdot)$, $\convp$ denotes the convergence in probability and $x\convleq y$ implies that the random variable $x$ is less than or equal to $y$ with probability 1. Furthermore, $a\sim \mathcal{CN}(\mu,\sigma^2)$ means that $a$ is a complex Gaussian random variable with mean $\mu$ and variance $\sigma^2$.

\section{System Model and Problem Formulation}
A multi-user wireless cellular network with a central base station at each cell is considered. For sake of analysis, we consider one individual cell and neglect the inter-cell interference. We assume that the base station in the considered cell has $N$ antennas and there are $M$ subcarriers available for the communication between the users and the base station. Each subcarrier is assigned to $K$ users simultaneously, thus the total number of users is $MK$. It is assumed that each user has $L$ antennas. 

Throughout the whole paper, we consider the uplink channel. It is assumed that the base station is able to do one-step SIC which is explained in Definition~\ref{definsucc} for a simple setting. 
\begin{definition}\label{definsucc}
Let $x_1$ and $x_2$ be the iid Gaussian signals of two single-antenna users with transmit power $P_1$ and $P_2$, respectively. Assume that a receiver with $N$ antennas knows the codebooks of both users and the received signal at the receiver is 
\begin{eqnarray}
\matr y=\matr{h_1}x_1+\matr{h_2}x_2+\matr n,
\end{eqnarray}
for some known channel vectors $\matr h_1,\matr h_2$ and the additive white Gaussian noise $\matr n\sim\mathcal{CN}(\matr 0,\sigma_n^2 \matr{I})$. The receiver can first detect and decode the signal of one of the users using a Minimum Mean Square Error (MMSE) detector and then cancel it from the received signal and detect the signal of the other user. We call this well-known detection method one-step SIC throughout this paper. From information theory, the maximum total rate is obtained as $\log_2(1+\frac{\|\matr h_1\|^2P_1+\|\matr h_2\|^2P_2}{\sigma_n^2})$ \cite{cover:91}. 
\end{definition}

In this paper, we limit the analysis to one-step SIC. One can consider SIC with higher steps at the expense of complexity.  
For given rates on the multiple-access channel with two users, the total transmit power is minimized if the base station first detects the user with the stronger channel and then the user with the weaker channel. 

The fundamental problem of NOMA is how to assign the available resources to the users. In the considered network, the available resources are the subcarriers and the spatial dimensions obtained from the MIMO nature of the system. 
We analyze such a problem in the SAMS setting in Sections~\ref{singledifffading} and \ref{samfadingsec} for frequency-selective and frequency-flat fading, respectively.
 In Sections~\ref{multiantennasing} and \ref{rankdef}, we address the MBASS-SAU and MBASS-MAU setting, respectively.

In this paper a single cell uplink channel is considered in which several users are distributed within the cell around the base station. To model the path loss of the users, similar to \cite{ding2016generalTWCOM}, we assume that the users are uniformly distributed in a disc with radius $R_d$ and the base station is at the center of the disc. The shadowing effect is neglected in this paper. The small scale fading coefficients are assumed to be iid Gaussian distributed, and the path loss of $k$th user is modeled as $\frac{1}{r_0^2+r_k^2}$ where $r_k$ is the distance between the $k$th user and the base station and $r_0$ is a constant to avoid singularity\footnote{In practice, $r_0$ is on the order of a wavelength and accounts for the fact that free-space path loss does not apply in the near-field of antennas.}. 

Although the proposed schemes in this paper can be applied to an uplink channel with any power allocation strategy, we limit the numerical and simulation results to two types of power control methods for the user terminals. One is when the users set their power proportional to the inverse of their path loss. In this case, the power multiplied by the path loss for all the users is equal to a constant. We call this method perfect power control (PPC). The second method is when the users transmit with the same power independent of their path loss which is called the equal power (EP) method in the sequel. Investigating optimum joint power allocation and user pairing in NOMA uplink analytically is an interesting future work.

\section{SAMS in Frequency-Selective Fading}
\label{singledifffading}
In this section, we consider the problem of assigning each of the $M$ subcarriers to two single-antenna users when the base station has a single antenna, i.e., $K=2$, $N=1$ and $L=1$. Thus, there are $2M$ users in total. Assume without loss of generality that the $\ell_{2m-1}$th and the $\ell_{2m}$th user are selected to transmit on the $m$th subcarrier. The received signal at the base station for the $m$th subcarrier is 
\begin{eqnarray}
y_m=\sqrt{d_{\ell_{2m-1}}}h_{\ell_{2m-1},m}x_{\ell_{2m-1},m}+\sqrt{d_{\ell_{2m}}}h_{\ell_{2m},m}x_{\ell_{2m}}+n_m,
\end{eqnarray}
where $d_{i}$, $h_{i,j}\sim\mathcal{CN}(0,1)$ and $x_{i,j}$ model the path loss, the small scale fading and the transmitted signal of the $i$th user at the $j$th subcarrier, respectively, and $n_m$ is additive white Gaussian noise with variance $\sigma_n^2$. It is assumed that the path loss of each user is identical at different subcarriers, but the fast fading coefficients differ. The signals of the users are assumed to be independent and Gaussian distributed. Let $P_i$ be the transmit power of the $i$th user which is selected prior to the subcarrier assignment according to a power allocation method. Furthermore, it is assumed that the base station knows the $d_i$s and $h_{i,j}$s perfectly. Note that in general, power and subcarrier allocation can be done jointly. However, in this paper we assume that the power allocation has been done beforehand based on the path loss of the users.

At the beginning of each transmission interval, the base station determines the optimum pairing based on the channel information. We assume that during each transmission interval $d_i$s and $h_{i,j}$s are constant. Thus, the transmission interval must be shorter than the coherence time of the channel. 

The base station applies one-step SIC to detect the signals with the sum rate 
\begin{eqnarray}
R_{m}(\ell_{2m-1},\ell_{2m})\define \log_2\left(1+\frac{|h_{\ell_{2m-1},m}|^2d_{\ell_{2m-1}}P_{\ell_{2m-1}}+|h_{\ell_{2m},m}|^2d_{\ell_{2m}}P_{\ell_{2m-1}}}{\sigma_n^2} \right), 
\end{eqnarray}
at the $m$th subcarrier.
For such a setting, the optimum user pairing is formulated as
\begin{eqnarray} \label{optNOMa2MAC}
\Pi_{\rm opt}=\argmax\limits_{(\ell_{2m-1},\ell_{2m})\in\Pi}\sum\limits_{m=1}^M R_m(\ell_{2m-1},\ell_{2m}),
\end{eqnarray} 
where $\Pi$ models all the possible pairing sets. $\Pi_{\rm opt}$ denotes the optimal pairing scheme which gives the maximum sum rate. 
\subsection{A sub-optimum algorithm with polynomial time complexity}
The total number of pairing sets is $\frac{(2M)!}{2^M}$ which can be approximated using Stirling's approximation as $\sqrt{4\pi M} \left(\frac{\sqrt{2}M}{\e}\right)^{2M}$. As observed, the total number of sets grows superexponential with $M$, thus the optimum solution is not feasible to calculate even for moderate $M$. To simplify the problem, we use the bounds provided in the following lemma.
\begin{lemma}\label{boundsrateMAC}
 $R_m$ is bounded by
\begin{eqnarray}
\frac{1}{2}\log_2\left(1+2\frac{|h_{\ell_{2m-1},m}|^2d_{\ell_{2m-1}}P_{\ell_{2m-1}}}{\sigma_n^2} \right)+\frac{1}{2}\log_2\left(1+2\frac{|h_{\ell_{2m},m}|^2d_{\ell_{2m}}P_{\ell_{2m}}}{\sigma_n^2} \right) \leq R_{m}(\ell_{2m-1},\ell_{2m}) \nonumber \\
 \leq 
\log_2\left(1+\frac{|h_{\ell_{2m-1},m}|^2d_{\ell_{2m-1}}P_{\ell_{2m-1}}}{\sigma_n^2} \right)+\log_2\left(1+\frac{|h_{\ell_{2m},m}|^2d_{\ell_{2m}}P_{\ell_{2m}}}{\sigma_n^2} \right).
\end{eqnarray}
\end{lemma}
\begin{proof}
The lower bound can be easily proven using the inequality of arithmetic and geometric means. In fact, the lower bound is achievable using TDMA and the upper bound is the case that there is no interference between the users. The bounds become tight in the low Signal to Noise Ratio (SNR) regime.
\end{proof}
Assuming that the $\ell_{2m-1}$th and the $\ell_{2m}$th user transmit at the $m$th subcarrier, the main advantage of the bounds provided in Lemma~\ref{boundsrateMAC} is that they follow the general form
\begin{eqnarray}
g_{\ell_{2m-1},m}+g_{\ell_{2m},m},
\end{eqnarray}
in which the contribution of the users in the sum rate can be decoupled. $g_{i,j}$ approximately represents the contribution of $j$th user in the sum rate at the $i$th subcarrier.
Using any of the bounds in Lemma~\ref{boundsrateMAC} as an approximation for $R_m(\ell_{2m-1},\ell_{2m})$, we can solve the optimization problem in \eqref{optNOMa2MAC} in polynomial time according to the following lemma. 

\begin{lemma}\label{lem2}
Assume the pairing problem 
\begin{eqnarray}\label{apprprobassign}
\max\limits_{(\ell_{2m-1},\ell_{2m})\in\Pi}\sum\limits_{m=1}^M g_{\ell_{2m-1},m}+g_{\ell_{2m},m},
\end{eqnarray}
where $m\in\{1,\cdots,M\}$, $\ell_{i} \in \{1,\cdots,2M\}$ and $\Pi$ models all the pairing sets.
The optimization problem can be solved optimally using the Hungarian algorithm with a polynomial time complexity. 
\end{lemma}
\begin{proof}
To convert the problem to a minimization problem, we define
\begin{eqnarray}
\hat{g}_{i,j}\define \max \limits_{k,n}g_{k,n}-g_{i,j}.
\end{eqnarray}
It can be easily shown that the solution pairs of the problem 
\begin{eqnarray}
\min\limits_{(\ell_{2m-1},\ell_{2m})\in\Pi}\sum\limits_{m=1}^M \hat{g}_{\ell_{2m-1},m}+\hat{g}_{\ell_{2m},m},
\end{eqnarray}
are the same as the solution pairs of the original problem. Furthermore, $\hat{g}_{i,j}\geq 0$. Next, 
define $\matr{Z}\in\mathbb{R}^{2M\times M }$ such that its $({i,j})$th entry is equal to $\hat{g}_{i,j}$.
To obtain the optimal pairs, we can solve an equivalent problem which is a linear assignment problem with the cost matrix 
$\matr{Z}_{\rm t}\define[\matr{Z} ~\matr{Z}]$.  The problem is the same problem as the job-worker problem in which we have $2M$ workers and $M$ jobs and we want to assign each job to two workers. To solve such a problem, we copy the cost matrix $\matr{Z}$ for $M$ new virtual jobs and minimize the cost by assigning each job to one worker exclusively. Finally, the workers assigned to the $i$th job and the $(i+M)$th virtual job will be assigned to $i$th job which means two workers per job. This linear assignment problem can be solved optimaly using the Hungarian algorithm with a polynomial time complexity \cite{kuhn1955hungarian}. The complexity of the Hungarian algorithm in this case is $\mathcal{O}(M^3)$ \cite{lawler2001combinatorial} 
\end{proof}

 Due to the lack of space here we omit to explain the Hungarian algorithm. To see the details of the Hungarian algorithm, some simple examples and also an online demo video, please see \cite{makohon2016hungarian}. Lemma~\ref{lem2} gives a sub-optimum algorithm for the user pairing problem. To derive the performance of this sub-optimum method, first we derive the optimum set to maximize the lower bound or the upper bound in Lemma~\ref{boundsrateMAC}, and then we calculate the sum rate for this given set. 


\subsection{A large system performance analysis for optimum pairing}
Define $\gamma_i\define P_id_i/\sigma_n^2$ which is the receive SNR of the $i$th user at the base station. Therefore, 
\begin{eqnarray}
R_m(\ell_{2m-1},\ell_{2m})= \log_2\left(1+\gamma_{\ell_{2m-1}} |h_{\ell_{2m-1},m}|^2+\gamma_{\ell_{2m}}|h_{\ell_{2m},m}|^2\right). 
\end{eqnarray}
For the sake of analysis, we consider iid Gaussian $h_{i,j}$ with variance $1$. Furthermore,
we assume that the users are distributed around the base station in a finite size area and all the users have nonzero $\gamma_i$ at the base station. The following theorem states the main result of this subsection.
\begin{theorem}\label{theor1}
Let $h_{i,j}$ be iid Gaussian distributed with variance of 1 and $c_1\leq\gamma_i\leq c_2$ for two finite constants $c_1$ and $c_2$. For $M\rightarrow \infty$, 
\begin{eqnarray}\label{limitnoma1}
 \frac{\max\limits_{(\ell_{2m-1},\ell_{2m})\in\Pi}\sum\limits_{m=1}^M R_m(\ell_{2m-1},\ell_{2m})}{M\log_2\log M}\convp 1.
\end{eqnarray}
\end{theorem}
 \begin{proof}
 The proof is given in Appendix~\ref{prooftheo1}.
 \end{proof}
 Based on the result of Theorem~\ref{theor1}, it is easy to show that 
\begin{eqnarray}
\lim\limits_{M\rightarrow \infty}\frac{\expect\max\limits_{(\ell_{2m-1},\ell_{2m})\in\Pi}\sum\limits_{m=1}^M R_m(\ell_{2m-1},\ell_{2m})}{M\log_2\log M}= 1,
\end{eqnarray}
where the expectation $\expect$ is over the channel realizations. 

Another interesting asymptotic behavior is the case of fixed $M$ and $\sigma_n^2\rightarrow 0$. One can use the two bounds presented in the proof of Theorem~\ref{theor1} and show that 
\begin{eqnarray}
 \frac{\max\limits_{(\ell_{2m-1},\ell_{2m})\in\Pi}\sum\limits_{m=1}^M R_m(\ell_{2m-1},\ell_{2m})}{M\log_2(1+\frac{c_3}{\sigma_n^2}\log M)}\convp 1,
\end{eqnarray}
for $\sigma_n^2\rightarrow 0$ and any arbitrary finite and positive constant $c_3$. We omit the proof here since it is quite similar to the proof of Theorem~\ref{theor1}.


%

 

 \section{SAMS in Frequency-Flat Fading}
\label{samfadingsec}
If the channels of all users are flat in frequency, the channels are identical for all subcarriers, i.e., $h_{i,j}=h_i$. We use the same setting as in Section~\ref{singledifffading}, i.e., $N=L=1$ and $K=2$. The goal is to obtain the best pairing set which maximizes the total rate
 \begin{eqnarray}
\Pi_{\rm opt}=\argmax\limits_{(\ell_{2m-1},\ell_{2m})\in\Pi}\sum\limits_{m=1}^M R_m,
 \end{eqnarray}
 where in this case
 \begin{eqnarray}
 R_{m}\define \log_2\left(1+\frac{|h_{\ell_{2m-1}}|^2d_{\ell_{2m-1}}P_{\ell_{2m-1}}+|h_{\ell_{2m}}|^2d_{\ell_{2m}}P_{\ell_{2m-1}}}{\sigma_n^2} \right). 
 \end{eqnarray}
 The following theorem presents the optimum user pairing strategy.
\begin{theorem}\label{optstrategysamefading}
If we sort the channel coefficients such that $| h_1|^2d_1P_1\geq |h_2|^2d_2P_2\geq \cdots\geq |h_{2M}|^2d_{2M}P_{2M}$, then the optimum pairing set which maximizes the total rate is
\begin{eqnarray}
\Pi_{\rm opt}=\{(1,2M),(2,2M-1),\ldots,(M,M+1)\}.
\end{eqnarray}
\end{theorem}
 \begin{proof}
 Assume the contrary, i.e.\ $\Pi=\{(1,2M),(2,2M-1),\ldots,(M,M+1)\}$ is not the optimum set. Then, there is at least one subset $\{w_1,w_2,w_3,w_4\}$ where $w_i\in\{1,\ldots,2M\}$ in the optimum set that is paired as one of the following options
 \begin{eqnarray}
 \Pi_1&=&\{(w_1,w_2),(w_3,w_4)\}, \\
  \Pi_2&=&\{(w_1,w_3),(w_2,w_4)\}. 
 \end{eqnarray}
 It is easy to show that these two pairing schemes result in a lower total rate compared to the paring scheme $\{(w_1,w_4),(w_2,w_3)\}$. This proves the theorem.
 \end{proof}
  
Next, we analyze the optimum pairing scheme for $\gamma_i\define \frac{d_iP_i}{\sigma_n^2}$ and $h_{i}\sim\mathcal{CN}(0,1)$ in the large system limit. In this case, we assume that the whole bandwidth including all the subcarriers is still smaller than the coherence bandwidth of the channel of the users such that the channel of each user is the same at different subcarriers. Let $f(\gamma)$ be the empirical density function of the SNRs $\gamma_i$.  

\begin{theorem}
For $h_{i}\sim\mathcal{CN}(0,1)$, $\gamma_i\sim f(\gamma_i)$ and ${M\rightarrow \infty}$, the average rate of the users converges to
\begin{eqnarray}
\bar{R}\define \frac{1}{2M}\sum\limits_{m=1}^M R_m\convp\int\limits_{0}^{1} \frac{1}{2}\log_2\left(1+F^{-1}(t/2) +F^{-1}(1-t/2)\right) \dif t,
\end{eqnarray}
where
\begin{eqnarray}\label{cdfz}
F(z)=1-\int_{\mathbb{R}} \e^{-z/x} f(x)\dif x,
\end{eqnarray}
and $F^{-1}(x)$ is the inverse of $F(x)$ with respect to composition. 
\end{theorem} 
\begin{proof}
Let $t_i\define \gamma_i|h_{i}|^2$ and assume that the samples are sorted as $t_1\geq t_2\geq \cdots\geq t_{2M}$. For large $M$, $t_i$ converges to $F^{-1}\left(\frac{i-1}{2M} \right)$ in probability where $F(\cdot)$ is the cumulative distribution function (cdf) of $\gamma_i|h_i|^2$ and $F^{-1}(\cdot)$ is the inverse with respect to composition. For $h_{i}\sim\mathcal{CN}(0,1)$ and $\gamma_i\sim f(\gamma_i)$, it can be shown that $F(z)$ can be calculated from \eqref{cdfz}.
Thus, for ${M\rightarrow \infty}$
\begin{eqnarray}
\bar{R}=\frac{1}{2M}\sum\limits_{m=1}^M R_m &=&\frac{1}{2M}\sum\limits_{m=1}^M \log_2\left(1+ t_{m}+ t_{2M-m+1}  \right)\nonumber \\
&\convp&\frac{1}{2M}\sum\limits_{m=1}^{M} \log_2\left(1+ F^{-1}\left(\frac{m-1}{2M} \right)+ F^{-1}\left(1-\frac{m}{2M} \right) \right)\nonumber \\
 &=&\int\limits_{0}^{1} \frac{1}{2}\log_2\left(1+F^{-1}(t/2) +F^{-1}(1-t/2)\right) \dif t.
\end{eqnarray}

\end{proof}
 For the PPC strategy, i.e, $\gamma_i=\gamma$, and $M\rightarrow \infty$, we have
 \begin{eqnarray}
 \frac{1}{2M}\sum\limits_{m=1}^M R_m\convp\int\limits_{0}^{1} \frac 12\log_2\left(1+\gamma\log\left( \frac{1}{1-t/2} \right)  +\gamma\log\left( \frac{2}{t} \right)\right) \dif t.
 \end{eqnarray}
 
As a benchmark, we also derive the average rate for random pairing. For large $M$, the average rate of random pairing converges to
\begin{eqnarray}
\int\limits_{0}^{\infty} \int\limits_{0}^{\infty}  \frac 12\log_2(1+t_1+t_2)g(t_1)g(t_2)\dif t_1 \dif t_2,
\end{eqnarray}
 where $g(t)$ is the probability density function of $t=\gamma_i|h_i|^2$ which is given as
 \begin{eqnarray}
 g(t)=\frac{\partial}{\partial t} F(t)=\int_{\mathbb{R}} \frac{\e^{-t/x}}{x} f(x)\dif x.
 \end{eqnarray}
Finally, in the following lemma, we prove that for the considered setting, the optimum pairing strategy for the max-min rate criterion is the same strategy as introduced in Theorem~\ref{optstrategysamefading}. The max-min rate criterion assures to maximize of the minimum sum rate of the subcarriers.
\begin{theorem}
For $| h_1|^2d_1P_1\geq |h_2|^2d_2P_2\geq\cdots\geq |h_{2M}|^2d_{2M}P_{2M}$, the optimum pairing set which is derived from
\begin{eqnarray}
\Pi_{\rm mm}=\argmax\limits_{\Pi} \min\limits_{m} R_m,
\end{eqnarray}
 is
\begin{eqnarray}
\Pi_{\rm mm}=\{(1,2M),(2,2M-1),\ldots,(M,M+1)\}.
\end{eqnarray}
\end{theorem}
 \begin{proof}
 The proof is similar to the proof of Theorem~\ref{optstrategysamefading}. 
 \end{proof}

\section{MBASS-SAU Scenario}
\label{multiantennasing}
In this section, we consider a general NOMA scheme in which each subcarrier is assigned to $K$ single-antenna users. We focus on one of the subcarriers and investigate the assignment of spatial dimensions. We consider the case of $N\leq K\leq 2N$. For such a setting there are at most $N$ spatial dimensions available. The base station divides the set of users into $K/2$ pairs and applies one-step SIC to each pair. Note that the resources in this scenario are the available spatial dimensions. 

For ease of notation, assume that $K$ is an even integer. 
We model the received signal at the base station as
\begin{equation}
\matr{y}=\matr{Hx}+\matr{n},
\end{equation}
where $\matr{y}\in\mathbb{C}^N$ is the received vector, $\matr{H}\in\mathbb{C}^{N\times K}$ is the channel matrix, $\matr{x}\in\mathbb{C}^K$ is the transmit vector of all the users and $n\sim \mathcal{CN}(\matr{0},\sigma_n^2\matr{I})$ is additive white Gaussian noise at the base station. Assume that $\matr{x}$ is iid Gaussian with covariance matrix $\matr\Phi$. 
 The base station first applies a $K\times N$ detection matrix
\begin{eqnarray}
\matr{r}=\matr{T}\matr{y}.
\end{eqnarray}
 This pairing is solely done at the base station and the users need not know the pairing strategy. 

\subsection{Optimum user pairing}
Assume that the $\ell_{2m-1}$th and the $\ell_{2m}$th user are paired for $m\in\{1,\ldots,K/2\}$. Without loss of generality assume $\ell_{2m-1}<\ell_{2m}$. 
Let $\matr{T}_{m}\in\mathbb{C}^{2\times N}$ consist of the $\ell_{2m-1}$th and the $\ell_{2m}$th row of $\matr{T}$. Furthermore, let $\matr{H}_{m}$ consist of the $\ell_{2m-1}$th and the $\ell_{2m}$th column of $\matr{H}$ and $\matr{H}_{\setminus m}$ is the matrix $\matr{H}$ when $\matr{H}_{m}$ is excluded from it.
The sum rate of the $\ell_{2m-1}$th and the $\ell_{2m}$th user is calculated as
\begin{eqnarray}
R_m\define R_{\ell_{2m-1}}+R_{\ell_{2m}}=\log_2\det\left(\matr{I}+ \matr{T}_m \matr{H}_{m}\matr\Phi_m\matr{H}_{m}^\dagger\matr{T}_m^\dagger \matr{\Theta}_m^{-1} \right),
\end{eqnarray}
where $\matr\Phi_m$ is the covariance matrix of a vector including the $\ell_{2m-1}$th and the $\ell_{2m}$th user,
\begin{eqnarray}
\matr{\Theta}_m\define \matr{T}_{m} \left(\matr{H}_{\setminus m}\matr\Phi_{\setminus m}\matr{H}_{\setminus m}^\dagger+\sigma_n^2 \matr{I}  \right)\matr{T}_{m}^\dagger,
\end{eqnarray}
 and $\matr\Phi_{\setminus m}$ is the covariance matrix of the input vector when the $\ell_{2m-1}$th and the $\ell_{2m}$th entry of it are excluded. 
 
The following theorem is used to obtain the maximum total rate of the users. 
\begin{theorem}\label{besteig}
Let $\matr X\in\mathbb{C}^{2\times K}$ and $\matr{A},\matr{B}\in\mathbb{C}^{K\times K}$ be two Hermitian matrices. Furthermore, let $\matr{A}$ be non-negative with rank $2$ and $\matr{B}$ be a full rank positive definite matrix. Then,
\begin{eqnarray}
\max\limits_{\matr{X}} \det \left(\matr{I}+\matr{X} \matr{A} \matr{X}^\dagger \left(\matr{X} \matr{B} \matr{X}^\dagger  \right)^{-1}  \right)=(1+\lambda_1)(1+\lambda_2),
\end{eqnarray}
where $\lambda_1$ and $\lambda_2$ are the nonzero eigenvalues of $\matr{B}^{-1/2} \matr{A} \matr{B}^{-1/2}$.
\end{theorem}
\begin{proof}
From eigenvalue decomposition, we have $\matr{B}^{-1/2}\matr{A}\matr{B}^{-1/2}=\matr{V}_1\matr{\Lambda}\matr{V}_1^\dagger$ where $\matr{V}_1$ is a unitary matrix and $\matr\Lambda$ is a diagonal matrix containing the eigenvalues. By defining $\matr{V}\define\matr{X}\matr{B}^{1/2}$ and $\matr{U}\define \matr{V}\matr{V}_1$, we have
\begin{eqnarray}\label{lemmadet}
\max\limits_{\matr{X}} \det \left(\matr{I}+\matr{X} \matr{A} \matr{X}^\dagger \left(\matr{X} \matr{B} \matr{X}^\dagger  \right)^{-1}  \right)
&=&\max\limits_{\matr{V}} \det \left(\matr{I}+\matr{V}\matr{B}^{-1/2} \matr{A} \matr{B}^{-1/2} \matr{V}^\dagger \left(\matr{V}\matr{V}^\dagger \right)^{-1}\right) \nonumber \\
&=&\max\limits_{\matr{V}} \frac{\det \left(\matr{V}\left(\matr{I}+ \matr{B}^{-1/2}\matr{A}\matr{B}^{-1/2}\right)\matr{V}^\dagger\right)}{\det \left( \matr{V} \matr{V}^\dagger  \right)} \nonumber\\
&=&\max\limits_{\matr{U}} \frac{\det \left(\matr{U}\left(\matr{I}+\matr{\Lambda}\right)\matr{U}^\dagger\right)}{\det \left( \matr{U} \matr{U}^\dagger  \right)} .
\end{eqnarray}
$\matr{B}^{-1/2}\matr{A}\matr{B}^{-1/2}$ has only two nonzero eigenvalues denoted by $\lambda_1$ and $\lambda_2$. Without loss of generality assume that these two nonzero eigenvalues are the first and the second eigenvalues. Let $\matr{U}= [\matr{U}_1 \matr{u}_2]$ where $\matr{U}_1\in \mathbb{C}^{2\times {K-1}}$ and $\matr{u}_2\in \mathbb{C}^{2\times 1}$. Using the matrix determinant lemma results in 
\begin{eqnarray}
 \frac{\det \left(\matr{U}\left(\matr{I}+\matr{\Lambda}\right)\matr{U}^\dagger\right)}{\det \left( \matr{U} \matr{U}^\dagger  \right)}=\frac{\det \left(\matr{U}_1\left(\matr{I}+\matr{\Lambda}_1\right)\matr{U}_1^\dagger\right)}{\det \left( \matr{U}_1 \matr{U}_1^\dagger  \right)} \frac{1+\matr{u}_2^\dagger \left(\matr{U}_1\left(\matr{I}+\matr{\Lambda}_1\right)\matr{U}_1^\dagger  \right)^{-1} \matr{u}_2}{1+\matr{u}_2^\dagger \left(\matr{U}_1\matr{U}_1^\dagger  \right)^{-1} \matr{u}_2} .
\end{eqnarray}
 The maximum eigenvalue of $\left(\matr{U}_1\left(\matr{I}+\matr{\Lambda}_1\right)\matr{U}_1^\dagger  \right)^{-1}$ is less than or equal to the maximum eigenvalue of $\left(\matr{U}_1\matr{U}_1^\dagger  \right)^{-1}$ \cite{bhatia1997matrix}. Therefore, the optimum solution for $\matr{u}_2$ is the all-zero vector. One can repeat this strategy for the last $K-3$ columns of $\matr{U}_1$ and show that only the first two columns of the optimum matrix $\matr{U}$ are nonzero and 
\begin{eqnarray}
\max\limits_{\matr{X}} \det \left(\matr{I}+\matr{X} \matr{A} \matr{X}^\dagger \left(\matr{X} \matr{B} \matr{X}^\dagger  \right)^{-1}  \right)=(1+\lambda_1)(1+\lambda_2).
\end{eqnarray}

\end{proof}
Using Theorem~\ref{besteig}, we obtain
\begin{eqnarray}\label{contrell}
\max\limits_{\matr{T}_m} R_m=\log_2(1+\lambda_{m,1})+\log_2(1+\lambda_{m,2}),
\end{eqnarray}
where $\lambda_{m,1}$ and $\lambda_{m,2}$ are the nonzero eigenvalues of the matrix
\begin{eqnarray}\label{matrixeigenphi}
\matr{\Omega}\define\left(\matr{H}_{\setminus m}\matr\Phi_{\setminus m}\matr{H}_{\setminus m}^\dagger+\sigma_n^2 \matr{I}  \right)^{-1/2} (\matr{H}_{m}\matr\Phi_m\matr{H}_{m}^\dagger) \left(\matr{H}_{\setminus m}\matr\Phi_{\setminus m}\matr{H}_{\setminus m}^\dagger+\sigma_n^2 \matr{I}  \right)^{-1/2}.
\end{eqnarray}
Finally, the maximum total rate is derived as
\begin{eqnarray} \label{sumrateportsM}
\sum_{m=1}^{K/2} R_m=\sum_{m=1}^{K/2} \log_2(1+\lambda_{m,1})+\log_2(1+\lambda_{m,2}).
\end{eqnarray}

To obtain the optimum pair, we form a $K\times K$ cost matrix whose $(i,j)$th entry denotes the contribution of the $i$th and $j$th users in the total rate which is given in \eqref{contrell} for the $\ell_{2m-1}$th and the $\ell_{2m}$th user. Note that the contribution of the $\ell_{2m-1}$th and the $\ell_{2m}$th user in the total rate should be independent from how the other users, i.e., all the users except the $\ell_{2m-1}$th and the $\ell_{2m}$th user, are grouped. If this condition does not hold, the problem cannot be modeled as a linear assignment problem with a cost matrix. To have such a condition, the eigenvalues $\lambda_{m,1}$ and $\lambda_{m,2}$ should be the same for all pairing schemes in which the $\ell_{2m-1}$th and the $\ell_{2m}$th user are paired.  It is straightforward to observe from \eqref{matrixeigenphi} that this condition is fulfilled.

Having the cost matrix, the problem becomes a linear assignment, and we find the optimum user pairing set using the Hungarian algorithm. The $(\ell_{2m-1},\ell_{2m})$th and the $(\ell_{2m},\ell_{2m-1})$th entries of the cost matrix are equal to $\log_2(1+\lambda_{m,1})+\log_2(1+\lambda_{m,2})$. The Hungarian algorithm gives us $K/2$ pairs with the maximum total rate.

Note that the scheme introduced in this section is not in power domain but in spatial domain. The pairing is done at the base station, and the users need not do anything. At the base station, the users are grouped into pairs and interference cancellation is applied to each pair. The approach introduced here can be used in massive MIMO uplink.

\subsection{Large system analysis}
In this subsection, we analyze the NOMA scheme presented in the previous subsection in the large system limit, i.e., $N,K\rightarrow \infty$ and constant $\alpha\define K/N$. The goal is to calculate the eigenvalues $\lambda_{m,1}$ and $\lambda_{m,2}$ in the large system limit. We confine the analysis to the case of PPC and leave the case of EP for future works. For PPC, we show that optimum user pairing performs identical to random pairing in the large system limit. However, in the results section, we demonstrate that for finite $K,N$, optimum user pairing performs much better than random pairing. Furthermore, in the results section, using computer simulation we show that the same behavior is observed for EP.



For PPC, we set $P_i=1/d_i$. Therefore, the path loss effect is canceled out. To analyze such a NOMA system, one can consider a similar NOMA system in which the users have identical power and the channel matrix represents the small scale fading effect. In the following theorem, we present the main results of this subsection.

\begin{theorem}\label{lemmafree}
Let $\matr{H}$ be an iid matrix whose entries have zero mean and variance of $1/N$. Let $\matr{\Phi}/\sigma_n^2=\gamma \matr{I}$. Then, for $K,N\rightarrow \infty$, constant $\alpha=K/N$, PPC, and any pairing method, we have
\begin{eqnarray}
 \lambda_{m,1},\lambda_{m,2}\convp
\sqrt{\frac{(1-\alpha)^2\gamma^2}{4}+\frac{(1+\alpha)\gamma}{2} +\frac{1}{4}}-\frac{1}{2}+\frac{(1-\alpha)\gamma}{2}.
\end{eqnarray}
\end{theorem}
\begin{proof}
The proof is given in Appendix~\ref{prooflemmafree}.
\end{proof}
Note that Theorem~\ref{lemmafree} is very general and independent of the distribution of the entries of $\matr{H}$. In fact, Theorem~\ref{lemmafree} results from channel hardening, i.e.\ in the large system limit of PPC there is no channel better than any other. From Theorem~\ref{lemmafree}, it is concluded that for $K,N\rightarrow \infty$, we have
\begin{eqnarray}
 \frac{1}{K}\sum_{m=1}^{K/2} R_m\convp  \frac{2}{\alpha}\log_2\left(\frac 12+ 
 \sqrt{\frac{(1-\alpha)^2\gamma^2}{4}+\frac{(1+\alpha)\gamma}{2} +\frac{1}{4}}+\frac{(1-\alpha)\gamma}{2}\right).
\end{eqnarray}
Theorem~\ref{lemmafree} shows that in the large system limit, the total rate of PPC converges to a limit independent of the pairing strategy. This limit is equal to the total rate of optimum linear detection and random pairing. However, in finite sizes, in the numerical results section we show that optimum user pairing achieves higher total rate than random pairing. We conjecture that the gain of optimal pairing in NOMA remains in the large system limit if the size of the groups scales with $K$. In such a case, interference cancellation with higher steps is required.

\section{MBASS-MAU Scenario}
\label{rankdef}

Consider the case that all user terminals and the base station have multi-antenna arrays. Each user has $L$ antennas and the number of antennas at the base station is $N$. It is assumed that all the $K$ users communicate with the base station over a single subcarrier. For such an uplink channel, a NOMA scheme using a signal alignment technique is proposed in \cite{ding2016generalTWCOM} in which the users are divided into $K/2$ pairs. Then, the inter-pair interference is mitigated using a combination of beamforming at the user terminals and linear detection at the base station if $K\leq 2N<4L$ and the channels between the users and the base station are full rank. The base station then applies one-step SIC for every pair. It is shown in \cite{ding2016generalTWCOM} that this NOMA system can be decomposed into $K/2$ orthogonal single-antenna NOMA systems.

In this section, we use the same method introduced in the last section to propose a new NOMA uplink scheme in which each user terminal uses a beamforming vector to transmit a single data stream. We show that this method has higher performance than the NOMA scheme with signal alignment proposed in \cite{ding2016generalTWCOM}. Furthermore, in contrast to the scheme with signal alignment, in our scheme users only need to know their own channel. 

Assume that the $k$th user sends the data stream $u_k$ by applying the beamforming vector $\matr b_k\in \mathbb{C}^{L\times 1}$. The received vector at the base station reads 
\begin{eqnarray}
\matr{y}=\sum\limits_{k=1}^K \matr{G}_{k}\matr{b}_ku_k+\matr{n},
\end{eqnarray}
where $\matr{G}_{k}$ is the channel matrix between the $k$th user and the base station.
In \cite{ding2016generalTWCOM}, full rank channel matrices and $2L>N$ are assumed to guarantee a user pairing solution. We show that none of these assumptions are required for the method introduced in this section. Note that we have already shown in the previous section that this NOMA method works even for single-antenna users.

It is assumed that the users only know their own channel and the base station knows all the channel matrices. We use a simple beamforming method at the user terminals and leave the optimum beamforming strategy for the future works. The $m$th user sets 
\begin{eqnarray}
\matr{b}_m=\matr{v}_m,
\end{eqnarray}
where $\matr{v}_m$ is the eigenvector of $\matr{G}_m^\dagger \matr{G}_m$ which corresponds to the maximum eigenvalue. Letting
\begin{eqnarray}
\breve{\matr{H}}\define[\matr{G}_1\matr{b}_1,\ldots,\matr{G}_K\matr{b}_{K}]
\end{eqnarray}
and 
\begin{eqnarray}
{\matr{u}}\define[u_1,\ldots,u_{K}]^\T,
\end{eqnarray}
the received signal at the base station is
\begin{eqnarray}\label{singlmodd}
\matr{y}=\breve{\matr{H}}{\matr{u}}+\matr{n}.
\end{eqnarray}
Eq. \eqref{singlmodd} describes the same channel model as in the case of single-antenna user terminals. Thus, we apply the same detection technique as described in the previous section to pair the users for the case of $K=2N$. The optimum user pairing set can be obtained using the Hungarian algorithm. Note that in this NOMA scheme, the channel matrices are not required to be full rank. Furthermore, $L$ can be any integer. In the numerical results section, we compare this method against the NOMA scheme with signal alignment.

\section{Numerical and simulation results}
\label{results}
In this section, some numerical results for the proposed NOMA schemes explained in the previous sections are presented. The cell is assumed to be a disc with radius of $R_d=100$ and $r_0$ is set to $1$. Both the PPC and EQ strategies are considered. We define the average transmit power of the users as
\begin{eqnarray}
\bar{P}\define\expect P_i,
\end{eqnarray}
and the average path loss of the users as
\begin{eqnarray}
\bar{d}\define\expect d_i.
\end{eqnarray}
To have a fair comparison, we use the parameter $\bar{\gamma}\define \frac{\bar{d}\bar{P}}{\sigma_n^2}$ as a measure for SNR at the receiver. Note that $\bar{\gamma}$ is different than $\expect \gamma_i$. We do not use the measure $\expect \gamma_i=\expect d_i P_i/\sigma_n^2$ for SNR since the power of the users appears with the weight of the path loss coefficients in it.

Besides the results which are obtained from analytical formulas, all the rates are derived using expectation over both the users' positions in the cell and the small scale fading coefficients. Note that the shadowing effect is neglected here. The rates are calculated using computer simulation with enough samples.

As a reference for comparison, an orthogonal method based on TDMA is used in which the $i$th user is assigned $\zeta_i$ portion of the total time at a given subcarrier. We consider two TDMA cases: 1) when the users at each subcarrier are assigned $\zeta_i=1/K$ portion of the total time, 2) when $\zeta_i$s are optimized for all the subcarriers.

To measure the fairness in different scenarios, we use the Jain index given by \cite{jain1984quantitative}
\begin{eqnarray}
\text{Jain's fairness index}=\frac{\left(\sum\limits_{k=1}^{MK} R_{k}   \right)^2}{MK \sum\limits_{k=1}^{MK} R_k^2},
\end{eqnarray}
which is between 0 and 1. Note that the maximum fairness is obtained when all the rates are equal.

\subsection{SAMS in frequency-selective fading}\label{resulselective}
We present the results for the sub-optimum user pairing method introduced in Section~\ref{singledifffading}. Note that the results here are lower bounds for the total rate, since we have used the suboptimum user pairing methods based on the bounds in Lemma~\ref{boundsrateMAC}.
Fig.~\ref{singleuserselectivedivi5_10_20} shows the total rate of the users normalized by $M\log_2(1+2\bar{\gamma}\log M)$ for $M=10$ versus $\bar{\gamma}$. It is observed that in both the EP and PPC cases, the proposed pairing methods perform much better than random pairing. The results for the pairing methods based on the first and the second bound in Lemma~\ref{boundsrateMAC} are almost the same. Thus, in the remaining figures, we only consider the method based on the first bound. It is also observed in Fig.~\ref{singleuserselectivedivi5_10_20} that the NOMA scheme with EP has a better performance than the one with PPC. The reason is that for PPC, the users compensate the path losses, hence the received powers of the users at the base station differ less and accordingly the gain of NOMA reduces. Another observation in Fig.~\ref{singleuserselectivedivi5_10_20} is that for the case of $\bar{\gamma}\rightarrow \infty$, the normalized total rate converges to 1 which confirms the predictions of Section~\ref{singledifffading}. 

\begin{figure}[t!]
\centering
\resizebox{.6\linewidth}{!}{
\pstool[width=.55\linewidth]{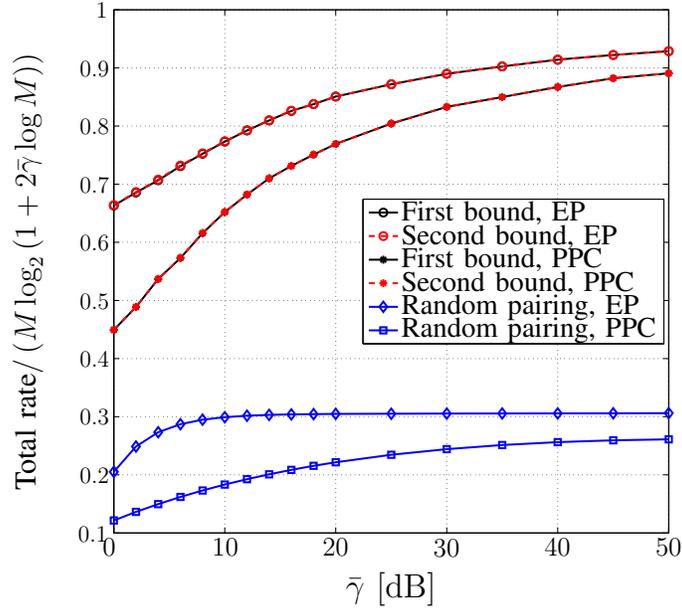}{
\psfrag{data1data1data1111datadata}[l][l][.85]{First bound, EP}
\psfrag{data2}[l][l][.85]{Second bound, EP}
\psfrag{data3}[l][l][.85]{First bound, PPC}
\psfrag{data4}[l][l][.85]{Second bound, PPC}
\psfrag{data5}[l][l][.85]{Random pairing, EP}
\psfrag{data6}[l][l][.85]{Random pairing, PPC}

\psfrag{Maximum total rate}[c][c][.9]{$\text{Total rate}/\left(M\log_2\left(1+2\bar{\gamma} \log M  \right) \right)$}
\psfrag{gamma}[c][c][1]{$\bar{\gamma}~[{\rm dB}]$}

\psfrag{0}[c][l][.8]{$0$}
\psfrag{5}[c][c][.8]{$5$}
\psfrag{10}[c][c][.8]{$10$}
\psfrag{15}[c][c][.8]{$15$}
\psfrag{20}[c][c][.8]{$20$}
\psfrag{30}[c][c][.8]{$30$}
\psfrag{40}[c][c][.8]{$40$}
\psfrag{50}[c][c][.8]{$50$}

\psfrag{0.1}[c][c][.7]{$0.1\hspace{1mm}$}

\psfrag{0.2}[c][c][.7]{$0.2\hspace{1mm}$}
\psfrag{0.3}[c][c][.7]{$0.3\hspace{1mm}$}
\psfrag{0.4}[c][c][.7]{$0.4\hspace{1mm}$}
\psfrag{0.5}[c][c][.7]{$0.5\hspace{1mm}$}
\psfrag{0.6}[c][c][.7]{$0.6\hspace{1mm}$}
\psfrag{0.7}[c][c][.7]{$0.7\hspace{1mm}$}
\psfrag{0.8}[c][c][.7]{$0.8\hspace{1mm}$}
\psfrag{0.9}[c][c][.7]{$0.9\hspace{1mm}$}
\psfrag{1}[c][c][.7]{$1\hspace{1mm}$}
\psfrag{1.1}[c][c][.7]{$1.1\hspace{1mm}$}
\psfrag{1.2}[c][c][.7]{$1.2\hspace{1mm}$}

}}
\caption{The total rate normalized by $M\log_2\left(1+2\bar{\gamma} \log M  \right) $ for SAMS in frequency-selective fading and $M=10$.}

\label{singleuserselectivedivi5_10_20}
\end{figure}

Next, we investigate the total rate versus the number of subcarriers. $\bar{\gamma}$ is set to $5~{\rm dB}$. The total rates normalized by $M\log_2\log M$ are plotted in Fig.~\ref{singleuserselectiveM}. 
The results in both cases converge to $1$ very slowly. Note that the low convergence rate is due to the terms containing $\log_2 \log M$ in the normalized total rate. The results for random pairing are also plotted. As observed, the normalized total rates for random pairing converge to $0$ since the total rates in this case scale with $M$. As for the previous figure, the EP method performs better than PPC. 
\begin{figure}[t!]
\centering
\resizebox{.6\linewidth}{!}{
\pstool[width=.55\linewidth]{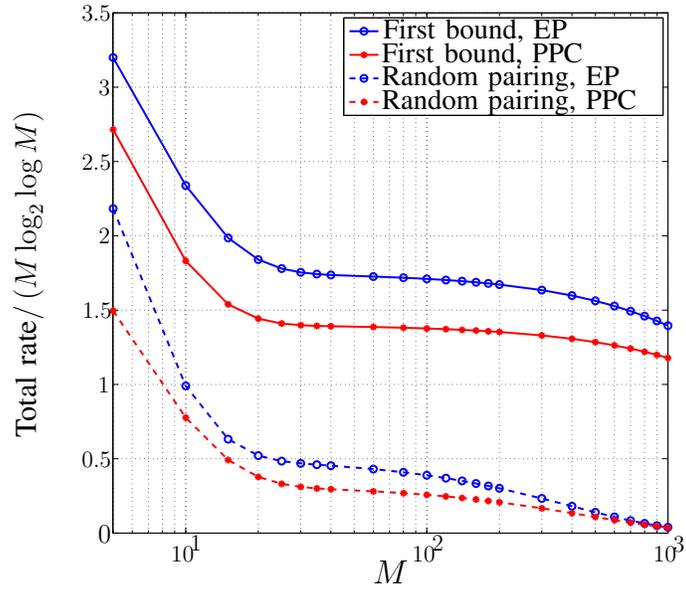}{
\psfrag{RLRLRRLRLRLRLRLRL1data1}[l][l][.85]{First bound, EP}
\psfrag{data2}[l][l][.85]{First bound, PPC}
\psfrag{data3}[l][l][.85]{Random pairing, EP}
\psfrag{data4}[l][l][.85]{Random pairing, PPC}
\psfrag{maximum total rate}[c][c][.9]{$\text{Total rate}/\left(M\log_2 \log M   \right)$}
\psfrag{M}[c][c][1]{$M$}

\psfrag{0}[c][c][.8]{$0$}
\psfrag{200}[c][c][.8]{$200$}
\psfrag{400}[c][c][.8]{$400$}
\psfrag{600}[c][c][.8]{$600$}
\psfrag{800}[c][c][.8]{$800$}
\psfrag{1000}[c][c][.8]{$1000$}
\psfrag{10}[c][c][.7]{$10$}

\psfrag{0.5}[c][c][.7]{$0.5$}
\psfrag{1}[c][c][.7]{$1$}
\psfrag{1.5}[c][c][.7]{$1.5$}
\psfrag{2}[c][c][.7]{$2$}
\psfrag{2.5}[c][c][.7]{$2.5$}
\psfrag{3}[c][c][.7]{$3$}
\psfrag{3.5}[c][c][.7]{$3.5$}
\psfrag{4}[c][c][.7]{$4$}
\psfrag{4.5}[c][c][.7]{$4.5$}
\psfrag{5}[c][c][.7]{$5$}
}}
\caption{The total rate normalized by $M\log_2\log M $ for SAMS in frequency-selective fading and $\bar{\gamma}=5$~dB versus $M$.}

\label{singleuserselectiveM}
\end{figure}

In Fig.~\ref{SelreqPowerDiscgamma5_10}, we compare the NOMA scheme with two OMA schemes based on TDMA which achieves the same set of rates. $M$ is set to $10$. The required total power normalized by the total power of the NOMA scheme is selected as performance measure. This measure shows how much power each method requires to reach the same rate sets as in the NOMA scheme. It is observed that the NOMA schemes are much more power efficient than the two OMA schemes. Note that the NOMA scheme does not need any time management which is very critical in the case of the TDMA method with optimum $\zeta_i$. Furthermore, it is observed that the TDMA methods require more power in the EP case. 

\begin{figure}[t!]
\centering
\resizebox{.6\linewidth}{!}{
\pstool[width=.58\linewidth]{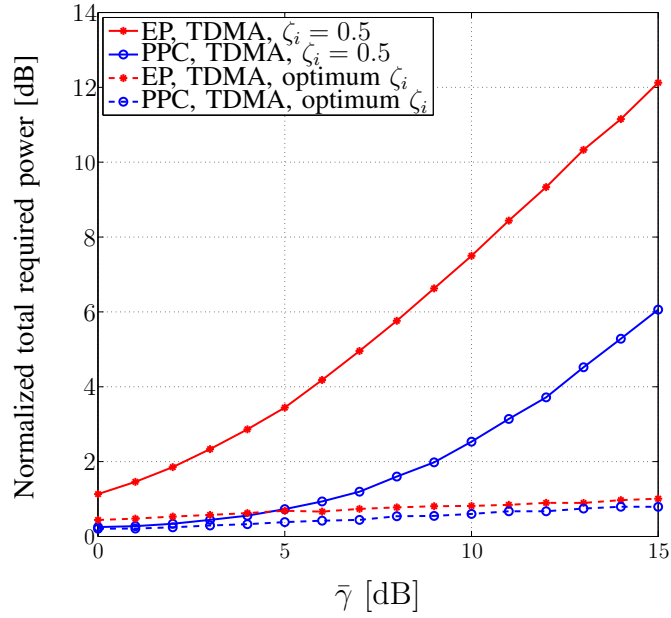}{
\psfrag{TDMA, time portion1portion1111}[l][l][.85]{EP, TDMA, $\zeta_i=0.5$}
\psfrag{data2}[l][l][.85]{PPC, TDMA, $\zeta_i=0.5$}
\psfrag{data3}[l][l][.85]{EP, TDMA, optimum $\zeta_i$}
\psfrag{data4}[l][l][.85]{PPC, TDMA, optimum $\zeta_i$}
\psfrag{Normalized total required power (dB)}[c][c][.9]{Normalized total required power [dB]}
\psfrag{gamma}[c][c][1]{$\bar{\gamma}~[{\rm dB}]$}

\psfrag{0}[c][c][.8]{$0$}
\psfrag{2}[c][c][.8]{$2$}
\psfrag{4}[c][c][.8]{$4$}
\psfrag{6}[c][c][.8]{$6$}
\psfrag{8}[c][c][.8]{$8$}
\psfrag{10}[c][c][.8]{$10$}
\psfrag{12}[c][c][.8]{$12$}
\psfrag{14}[c][c][.8]{$14$}

\psfrag{0}[c][c][.7]{$0$}
\psfrag{5}[c][c][.7]{$5$}
\psfrag{10}[c][c][.7]{$10$}
\psfrag{15}[c][c][.7]{$15$}
\psfrag{20}[c][c][.7]{$20$}

}}
\caption{The normalized total required power of two OMA schemes for SAMS in frequency-selective fading and $M=10$. The powers at each case are normalized with the total power of the corresponding NOMA scheme.}

\label{SelreqPowerDiscgamma5_10}
\end{figure}

Next, we plot the Jain fairness index versus the number of sub-carriers for different scenarios. The results are shown in Fig.~\ref{fairness_selectivefading}. As observed, the proposed pairing method performs much better in terms of fairness than random pairing method. Furthermore, PPC performs better than EP which is in fact clear because in PPC the path loss is compensated by power control. The main message of this figure is that our proposed user pairing method not only increases the sum rate but also improves fairness in comparison to random pairing.

\begin{figure}[t!]
\centering
\resizebox{.6\linewidth}{!}{
\pstool[width=.58\linewidth]{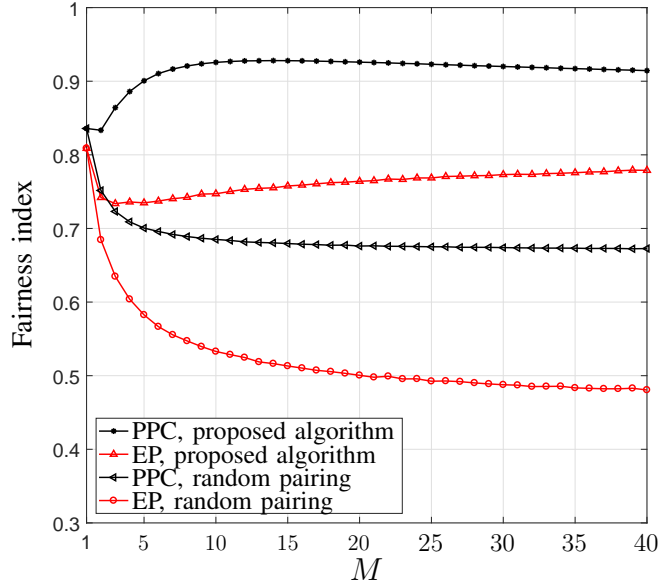}{
\psfrag{PPC, proposed algorithmalgorithm}[l][l][.8]{PPC, proposed algorithm}
\psfrag{EP, proposed algorithm}[l][l][.8]{EP, proposed algorithm}
\psfrag{PPC, random pairing}[l][l][.8]{PPC, random pairing}
\psfrag{EP, random pairing}[l][l][.8]{EP, random pairing}
\psfrag{Fairness index}[c][c][.9]{Fairness index}
\psfrag{M}[c][c][1]{$M$}

\psfrag{}[c][c][.8]{$1$}
\psfrag{5}[c][c][.8]{$5$}
\psfrag{10}[c][c][.8]{$10$}
\psfrag{15}[c][c][.8]{$15$}
\psfrag{20}[c][c][.8]{$20$}
\psfrag{25}[c][c][.8]{$25$}
\psfrag{30}[c][c][.8]{$30$}
\psfrag{35}[c][c][.8]{$35$}
\psfrag{40}[c][c][.8]{$40$}

\psfrag{0}[c][c][.7]{$0$}
\psfrag{5}[c][c][.7]{$5$}
\psfrag{10}[c][c][.7]{$10$}
\psfrag{15}[c][c][.7]{$15$}
\psfrag{20}[c][c][.7]{$20$}

}}
\caption{The Jain fairness index versus the number of sub-carriers in SAMS for PPC and EP power control methods.}

\label{fairness_selectivefading}
\end{figure}

\subsection{SAMS in frequency-flat fading}
In this subsection, the results for the case of frequency-flat channels in Section~\ref{samfadingsec} are presented. We compare the performance of optimum pairing against random pairing. The results are shown in Fig.~\ref{flatgamma5_10} for $\bar{\gamma}=15$~dB and the EP case. The simulation results together with the analytical results $M\rightarrow \infty$ are plotted. It is observed that optimum pairing has better performance than random pairing. As a benchmark, the results for full interference cancellation are also plotted which is, in fact, the ergodic capacity per user. In this method, all the users transmit on all subcarriers and the base station uses full interference cancellation. It is easy to show that in such a case for $M\rightarrow \infty$, the ergodic capacity per user converges to $\frac{1}{2}\log_2(1+2\bar{\gamma})\approx 3$. From Fig.~\ref{flatgamma5_10}, the rate per user of NOMA with optimum pairing is worse than the ergodic capacity per user. Note, however, that NOMA with one-step interference cancellation has much lower complexity than full interference cancellation. 
\begin{figure}[t!]
\centering
\resizebox{.6\linewidth}{!}{
\pstool[width=.59\linewidth]{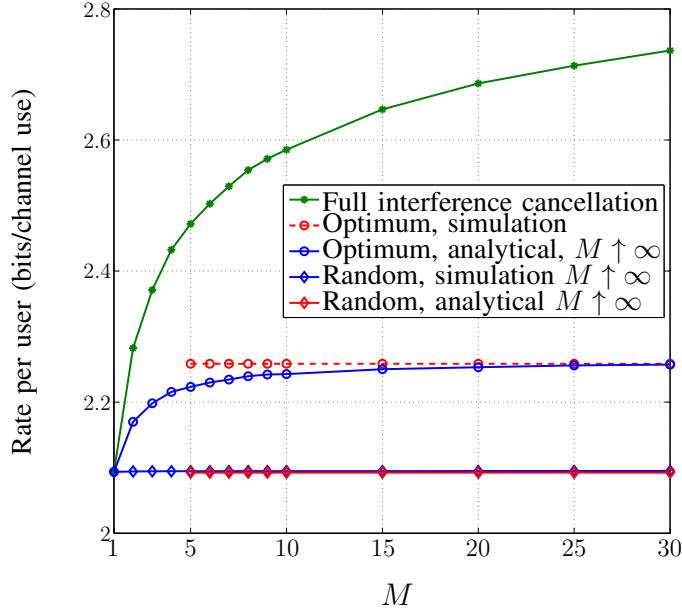}{
\psfrag{Analytical, optimum, gam11gamma10}[l][l][.9]{Full interference cancellation}
\psfrag{data3}[l][l][.9]{Optimum, analytical, $M\uparrow \infty$}
\psfrag{data2}[l][l][.9]{Optimum, simulation}
\psfrag{data4}[l][l][.9]{Random, simulation $M\uparrow \infty$}
\psfrag{data5}[l][l][.9]{Random, analytical $M\uparrow \infty$}

\psfrag{Total rate over M}[c][c][.9]{Rate per user (bits/channel use)}
\psfrag{M}[c][c][1]{$M$}

\psfrag{0}[c][c][.8]{$0$}
\psfrag{1}[c][c][.8]{$1$}
\psfrag{5}[c][c][.8]{$5$}
\psfrag{10}[c][c][.8]{$10$}
\psfrag{15}[c][c][.8]{$15$}
\psfrag{20}[c][c][.8]{$20$}
\psfrag{25}[c][c][.8]{$25$}
\psfrag{30}[c][c][.8]{$30$}

\psfrag{2}[c][c][.7]{$2$}
\psfrag{2.2}[c][c][.7]{$2.2$}
\psfrag{2.4}[c][c][.7]{$2.4$}
\psfrag{2.6}[c][c][.7]{$2.6$}
\psfrag{2.8}[c][c][.7]{$2.8$}

}}
\caption{The rate per user of optimum and random user pairing for SAMS in frequency-flat fading, the EP method and $ \bar{\gamma}=15$~dB. The ergodic capacity per user, which is achieved when the full interference cancellation is used, is also plotted.}

\label{flatgamma5_10}
\end{figure}

Next, we plot the analytical results for $M\rightarrow \infty$ versus $\bar{\gamma}$ for optimum and random pairing in both the EP and PPC case. The results are given in Fig.~\ref{analyti_flat}. It is observed that optimum pairing performs about 1~dB better than random pairing for both EP and PPC. The EP method performs about 2.5~dB better than the PPC method at around $\bar{\gamma}=10$~dB. As a benchmark, the ergodic capacities of EP and of PPC are also plotted.  The EP method with optimum pairing looses about 4~dB compared to the ergodic capacity. For the PPC method, optimum pairing falls only little behind the ergodic capacity.

\begin{figure}[t!]
\centering
\resizebox{.6\linewidth}{!}{
\pstool[width=.66\linewidth]{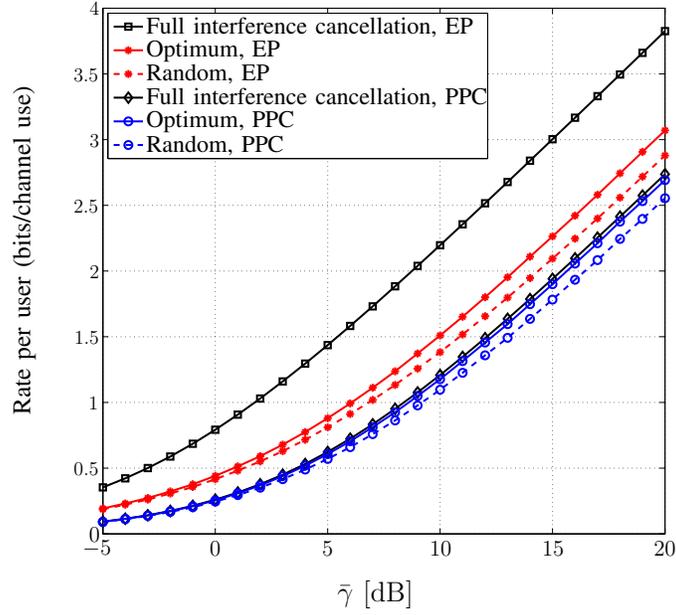}{
\psfrag{Rateperuserbitschanneuse11}[l][l][.9]{Rate per user (bits/channel use)}
\psfrag{data1data1111data1111data1data111}[l][l][.85]{Full interference cancellation, EP}
\psfrag{data2}[l][l][.85]{Optimum, EP}
\psfrag{data3}[l][l][.85]{Random, EP}
\psfrag{data4}[l][l][.85]{Full interference cancellation, PPC}
\psfrag{data5}[l][l][.85]{Optimum, PPC}
\psfrag{data6}[l][l][.85]{Random, PPC}

\psfrag{gammabar}[l][l][1]{$\bar{\gamma}~[{\rm dB}]$}

\psfrag{-10}[c][c][.8]{$-10$}
\psfrag{-5}[c][c][.8]{$-5$}
\psfrag{0}[c][c][.8]{$0$}
\psfrag{5}[c][c][.8]{$5$}
\psfrag{10}[c][c][.8]{$10$}
\psfrag{15}[c][c][.8]{$15$}
\psfrag{20}[c][c][.8]{$20$}
\psfrag{25}[c][c][.8]{$25$}
\psfrag{30}[c][c][.8]{$30$}

\psfrag{0.5}[c][c][.7]{$0.5$}
\psfrag{1}[c][c][.7]{$1$}
\psfrag{1.5}[c][c][.7]{$1.5$}
\psfrag{2}[c][c][.7]{$2$}
\psfrag{2.5}[c][c][.7]{$2.5$}
\psfrag{3}[c][c][.7]{$3$}
\psfrag{3.5}[c][c][.7]{$3.5$}
\psfrag{4}[c][c][.7]{$4$}

}}
\caption{The average rate of the users versus $\bar{\gamma}$ for optimum and random pairing methods, and SAMS in frequency-flat fading.}

\label{analyti_flat}
\end{figure}

We also compare the performance of the NOMA scheme using optimum user pairing against the TDMA schemes in Fig.~\ref{reqPowerDiscflat}. The number of subcarriers is set to $M=10$. The normalized required total power to achieve the same set of rates as in the NOMA scheme are plotted versus SNR. The TDMA schemes with $\zeta_i=0.5$ require much higher power compared to the case of NOMA. This extra power for EP is much larger than for PPC.

\begin{figure}[t!]
\centering
\resizebox{.6\linewidth}{!}{
\pstool[width=.61\linewidth]{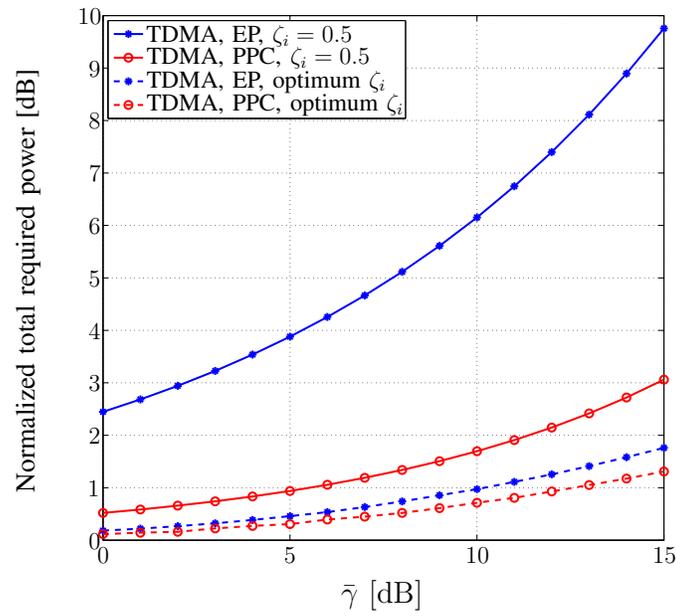}{
\psfrag{TDMA, time portion1portion1}[l][l][.8]{TDMA, EP, $\zeta_i=0.5$}
\psfrag{data2}[l][l][.8]{TDMA, PPC, $\zeta_i=0.5$}
\psfrag{data3}[l][l][.8]{TDMA, EP, optimum $\zeta_i$}
\psfrag{data4}[l][l][.8]{TDMA, PPC, optimum $\zeta_i$}

\psfrag{Normalized required power}[c][c][.9]{Normalized total required power [dB]}

\psfrag{gamma}[c][c][1]{$\bar{\gamma}~[{\rm dB}]$}

\psfrag{0}[c][c][.8]{$0$}
\psfrag{1}[c][c][.8]{$1$}
\psfrag{2}[c][c][.8]{$2$}
\psfrag{3}[c][c][.8]{$3$}
\psfrag{4}[c][c][.8]{$4$}
\psfrag{5}[c][c][.8]{$5$}
\psfrag{6}[c][c][.8]{$6$}
\psfrag{7}[c][c][.8]{$7$}
\psfrag{8}[c][c][.8]{$8$}
\psfrag{9}[c][c][.8]{$9$}

\psfrag{10}[c][c][.8]{$10$}
\psfrag{12}[c][c][.8]{$12$}
\psfrag{14}[c][c][.8]{$14$}
\psfrag{15}[c][c][.8]{$15$}
\psfrag{18}[c][c][.8]{$18$}
\psfrag{20}[c][c][.8]{$20$}

}}
\caption{The total required power of two OMA schemes normalized by the total power of the NOMA scheme for SAMS in frequency-flat fading and $M=10$.}

\label{reqPowerDiscflat}
\end{figure}

\subsection{MBASS-SAU}
The results for optimum user pairing in the case of multi-antenna base stations and one subcarrier are presented in this subsection. The users are assumed to have a single-antenna. The results are plotted in Fig.~\ref{multipleantenna_fig7gamma5_10} for $K/N=2$ and $\bar{\gamma}=15$~dB. The rate per user for random pairing is also plotted. It is observed that for small and moderate $N$,  optimum user pairing outperforms random pairing significantly. For $K,N\rightarrow \infty$ and both EP and PPC, the results for both optimum and random pairing converge to the same limit. For PPC, the analytical result for $M\rightarrow \infty$, derived based on Theorem~\ref{lemmafree} for $\alpha=2$, is also plotted. It is observed that the results for both random and optimum pairing converge to the analytical limit for large $N$.

\begin{figure}[t!]
\centering
\resizebox{.6\linewidth}{!}{
\pstool[width=.59\linewidth]{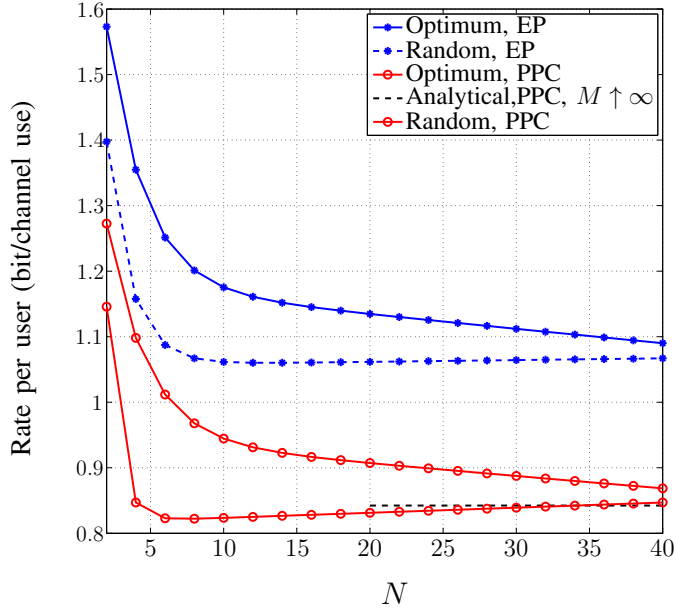}{
\psfrag{data1data1data1data1data1}[l][l][.8]{Optimum, EP}
\psfrag{data2}[l][l][.8]{Random, EP}
\psfrag{data3}[l][l][.8]{Optimum, PPC}
\psfrag{data4}[l][l][.8]{Analytical,PPC, $M\uparrow \infty$}
\psfrag{data5}[l][l][.8]{Random, PPC}

\psfrag{Rateperuserbitsperusers}[c][c][.9]{Rate per user (bit/channel use)}

\psfrag{N}[c][c][1]{$N$}

\psfrag{2}[c][c][.8]{$2$}
\psfrag{5}[c][c][.8]{$5$}
\psfrag{10}[c][c][.8]{$10$}
\psfrag{15}[c][c][.8]{$15$}
\psfrag{20}[c][c][.8]{$20$}
\psfrag{25}[c][c][.8]{$25$}
\psfrag{30}[c][c][.8]{$30$}
\psfrag{35}[c][c][.8]{$35$}
\psfrag{40}[c][c][.8]{$40$}

\psfrag{0.7}[c][c][.7]{$0.7$}
\psfrag{0.8}[c][c][.7]{$0.8$}
\psfrag{0.9}[c][c][.7]{$0.9$}
\psfrag{1}[c][c][.7]{$1$}
\psfrag{1.1}[c][c][.7]{$1.1$}
\psfrag{1.2}[c][c][.7]{$1.2$}
\psfrag{1.3}[c][c][.7]{$1.3$}
\psfrag{1.4}[c][c][.7]{$1.4$}
\psfrag{1.5}[c][c][.7]{$1.5$}
\psfrag{1.6}[c][c][.7]{$1.6$}

}}
\caption{The rate per user for MBASS-SAU. The base station has $N$ antenna and the users have single antenna. $\bar{\gamma}$ is set to 15~dB.}

\label{multipleantenna_fig7gamma5_10}
\end{figure}


\subsection{MBASS-MAU}
In this subsection, we present the results for the NOMA scheme introduced in Section~\ref{rankdef} in which the users and the base station are equipped with multiple antennas. The results are plotted for $K=2N$, $N=2L+1$ and $\bar{\gamma}=10$~dB in Fig.~\ref{multian_users8} for both EP and PPC. The results of NOMA with user alignment, random user pairing and optimum detection are also plotted for sake of comparison. We leave optimum user pairing in user alignment based NOMA for future works. It is observed that the method introduced in Section~\ref{rankdef} with random user pairing performs better than the alignment method with random pairing. Furthermore, optimum user pairing derived by the Hungarian algorithm has much better performance than random user pairing in the proposed NOMA scheme for finite $N$. It can be seen that similar to the case of single-antenna users, random user pairing approaches the performance of optimum user pairing for $N\rightarrow \infty$.

\begin{figure}[t!]
\centering
\resizebox{.6\linewidth}{!}{
\pstool[width=.62\linewidth]{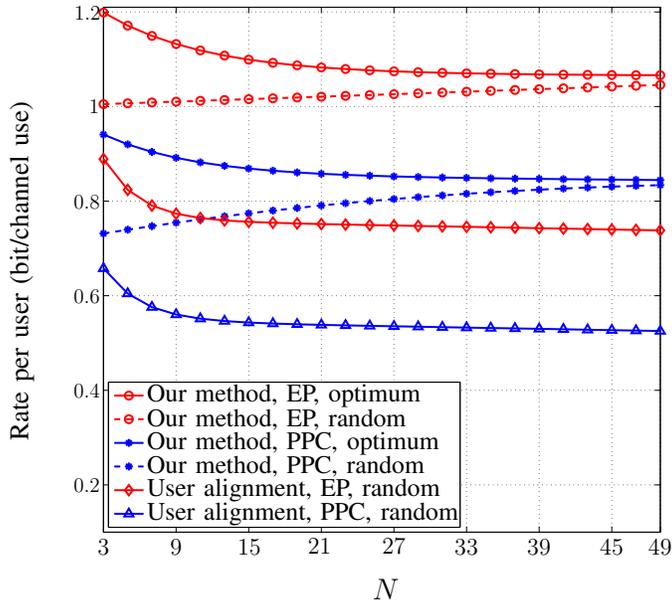}{
\psfrag{data1data1data1data1data11data1}[l][l][.85]{Our method, EP, optimum }
\psfrag{data2}[l][l][.85]{Our method, EP, random }
\psfrag{data3}[l][l][.85]{Our method, PPC, optimum }
\psfrag{data4}[l][l][.85]{Our method, PPC, random }
\psfrag{data5}[l][l][.85]{User alignment, EP, random }

\psfrag{data6}[l][l][.85]{User alignment, PPC, random}

\psfrag{ratebar}[c][c][.9]{Rate per user (bit/channel use)}

\psfrag{N}[c][c][1]{$N$}

\psfrag{3}[c][c][.8]{$3$}
\psfrag{9}[c][c][.8]{$9$}
\psfrag{15}[c][c][.8]{$15$}
\psfrag{21}[c][c][.8]{$21$}
\psfrag{27}[c][c][.8]{$27$}
\psfrag{33}[c][c][.8]{$33$}
\psfrag{39}[c][c][.8]{$39$}
\psfrag{45}[c][c][.8]{$45$}
\psfrag{49}[c][c][.8]{$49$}

\psfrag{0.2}[c][c][.7]{$0.2$}
\psfrag{0.4}[c][c][.7]{$0.4$}
\psfrag{0.6}[c][c][.7]{$0.6$}
\psfrag{0.8}[c][c][.7]{$0.8$}
\psfrag{1}[c][c][.7]{$1$}
\psfrag{1.2}[c][c][.7]{$1.2$}
\psfrag{0.95}[c][c][.7]{$0.95$}
\psfrag{1}[c][c][.7]{$1$}
\psfrag{1.05}[c][c][.7]{$1.05$}

}}
\caption{The average rate for MBASS-MAU, $N=2L+1=K/2$ and $\bar{\gamma}=10~{\rm dB}$.}

\label{multian_users8}
\end{figure}

\section{Conclusions and future works}
In this paper, NOMA uplink was considered and optimal user pairing was discussed in various settings. For single-antenna base stations and user terminals communicating over $M$ subcarriers with optimal user pairing, it was shown that the total rate scales with $M\log_2\log M$ if the users have iid small scale fading coefficients. For frequency-flat channels over the subcarriers, the optimal user pairing method was analyzed in the large system limit and shown to be a bit better than random pairing. In both cases, two orthogonal methods were compared against NOMA, and it was shown that the power efficiency of NOMA is much higher than the efficiencies of the orthogonal schemes. Furthermore, a novel NOMA scheme for the case of multi-antenna base stations was proposed in which the user terminals transmit over a single subcarrier. It was shown that optimum user pairing can be implemented by means of the Hungarian algorithm with polynomial complexity. The results for multi-antenna users were compared against NOMA based on signal alignment. It was shown that a significant performance gain can be achieved. 

In this paper, the power allocation was assumed to be fixed for the user terminals. Extending the results of this paper to the case of joint user pairing and power allocation can be interesting future work. One possible approach is to iterate user pairing and power allocation. Another problem for future is the joint assignment of subcarrier and spatial dimension. Furthermore, the methods in this paper were based on perfect channel state information. Investigating the effect of imperfect channel state information and extending the results to the case of user pairing based on large scale fading coefficients for the sake of complexity reduction can be interesting future works.

\begin{appendices}

\section{Proof of Theorem~\ref{theor1}}\label{prooftheo1}
We first introduce a lower bound and an upper bound for the maximum total rate and then prove the theorem using these two bounds.
\subsection{The upper bound}
A simple upper bound for the total rate is obtained when the best two users at each subcarrier are selected. Note that this method may not be achievable since some of the users may be among the best users for several subcarriers. Consider the $m$th subcarrier and let also $\kappa_{2m-1}$ and $\kappa_{2m}$ be the indices of the users with the first and second strongest channels on the $m$th subcarrier, i.e.,  
\begin{eqnarray}
\kappa_{2m-1}=\argmax\limits_{i\in\{1,\ldots,2M\}}\gamma_i |h_{i,m}|^2,
\end{eqnarray}
and
\begin{eqnarray}
\kappa_{2m}=\argmax\limits_{i\in\{1,\ldots,2M\},~i\neq \kappa_{2m-1}  }\gamma_i |h_{i,m}|^2.
\end{eqnarray}
Let the support of $\gamma_i$s be $[c_1,c_2]$ where $0<c_1<c_2<\infty$.
\begin{lemma}
For $M\rightarrow \infty$, 
\begin{eqnarray}
c_1\convleq \frac{\gamma_{\kappa_{2m-1}}|h_{\kappa_{2m-1},m}|^2}{\log(M)},\frac{\gamma_{\kappa_{2m}}|h_{\kappa_{2m},m}|^2}{\log(M)}\convleq c_2.
\end{eqnarray} 
\end{lemma}
\begin{proof}

The proof is straightforward using the Fisher-Tippett-Gnedenko theorem \cite{david1981order} and the fact that $c_1\leq \gamma_i\leq c_2$ for all $i\in\{1,\ldots,2M\}$.
\end{proof}
Using the result of this lemma, it is concluded that the total rate of this upper bound method is less than or equal to $M\log_2(1+2c_2 \log M)$ with probability $1$. Since this method gives an upper bound for the optimum total rate, it is concluded that for $M\rightarrow \infty$
\begin{eqnarray}
 \frac{\max\limits_{(\ell_{2m-1},\ell_{2m})\in\Pi}\sum\limits_{m=1}^K R_m}{M\log_2(1+2c_2 \log M)}\convleq 1.
\end{eqnarray}

\subsection{The lower bound}
Consider an achievable strategy in which for each subcarrier the first and the second best user among all the users excluding the users already selected for the previous subcarriers, are chosen. This means that we first select the two best users for the first subcarrier and then excluding this two users we select the best two users for the second subcarrier and continue this procedure up to the end. In fact, for the $m$th subcarrier, we select the best two users among $2(M-m+1)$ users.  Let $\eta_{2m-1}$ and $\eta_{2m}$ denote the indices of the first and the second selected user, respectively, at the $m$th subcarrier. Using a similar method as in the previous subsection, it can be shown that for all $2(M-m+1)$ where $m=\nu M$ and $\nu\in (0,1]$ is a non-vanishing constant, both  
$\frac{\gamma_{\eta_{2m-1}}|h_{\eta_{2m-1},m}|^2}{\log(M-m+1)}$ and $ \frac{\gamma_{\eta_{2m}}|h_{\eta_{2m},m}|^2}{\log(M-m+1)}$ are in $[c_1,c_2]$ with probability of $1$ when $M\rightarrow \infty$ \footnote{ Note that for finite $2(M-m+1)$, i.e., vanishing $\nu$, the statement may not be true but this does not change the asymptotic behavior.}. Therefore, the asymptotic behavior of the total rate of the considered method behaves similarly as $\sum\limits_{m=1}^M \log_2(1+2\hat{c} \log m)$ where $\hat{c}$ is a nonzero positive constant. Next, we present a lemma to investigate the asymptotic behavior of this sum:

\begin{lemma}\label{lemmastirtli}
For finite and positive $\hat{c}$, we have
\begin{eqnarray}
\lim\limits_{M\rightarrow \infty}\frac{\sum\limits_{m=1}^M \log_2(1+2\hat{c} \log m)}{M\log_2(1+2\hat{c}\log M)}= 1.
\end{eqnarray}
\end{lemma}
\begin{proof}
The function 
\begin{eqnarray}
w(x)=\frac{\log_2(1+a_1x)}{\log_2(1+a_2 x)}
\end{eqnarray}
is monotonically increasing if $a_1<a_2$ and decreasing if $a_1>a_2$ for $x\geq 0$. For $a_1=a_2$, the function becomes constant. Therefore, using L'Hospital's rule and the Stirling approximation, we have
\begin{eqnarray} \label{limit1}
\lim\limits_{M\rightarrow \infty}\frac{\sum\limits_{m=1}^M \log_2(1+2\hat{c} \log m)}{M\log_2(1+2\hat{c} \log M)}&\geq &
\lim\limits_{M\rightarrow \infty}\lim\limits_{\hat{c}\rightarrow 0^+}\frac{\sum\limits_{m=1}^M \log_2(1+2\hat{c} \log m)}{M\log_2(1+2\hat{c} \log M)}=
\lim\limits_{M\rightarrow \infty}\frac{\sum\limits_{m=1}^M \log m}{M \log M} \nonumber \\
&=&\lim\limits_{M\rightarrow\infty}\frac{\log(M!)}{M \log M}=1.
\end{eqnarray}
Furthermore, since $m\leq M$ in the summation, the limit is less than or equal to $1$. Therefore, the limit is $1$.
\end{proof}

Therefore, the total rate of the considered pairing method normalized by $M\log_2 \log M$ converges to $1$ in probability.

Finally, based on the results obtained from the above two subsections, it is concluded that 
\begin{eqnarray}\label{limitnoma2}
 \frac{\max\limits_{(\ell_{2m-1},\ell_{2m})\in\Pi}\sum\limits_{m=1}^M R_m}{M\log_2 \log M}\convp 1
\end{eqnarray}
 for $M\rightarrow \infty$. 

\section{Proof of Theorem~\ref{lemmafree}} \label{prooflemmafree}
For such a setting, we have
\begin{eqnarray}
\matr{\Omega}= \left(\matr{H}_{\setminus m}\matr{H}_{\setminus m}^\dagger+ \matr{I}/\gamma  \right)^{-1/2}  \matr{H}_{m}\matr{H}_{m}^\dagger \left( \matr{H}_{\setminus m}\matr{H}_{\setminus m}^\dagger+ \matr{I}/\gamma  \right)^{-1/2}.
\end{eqnarray}
The two nonzero eigenvalues of $\matr\Omega$ are equal to the eigenvalues of 
\begin{eqnarray}
\breve{\matr{\Omega}}= \matr{H}_{m}^\dagger \left( \matr{H}_{\setminus m}\matr{H}_{\setminus m}^\dagger+ \matr{I}/\gamma  \right)^{-1}  \matr{H}_{m},
\end{eqnarray}
which is a $2\times 2$ matrix. For $N,K\rightarrow \infty$, the off-diagonal entries of $\breve{\matr{\Omega}}$ converge to zero due to the fact that the entries of $\matr{H}_m$ and $\matr{H}_{\setminus m}$ are iid and zero mean. Furthermore, 
the two diagonal entries of $\breve{\matr{\Omega}}$ fulfill \cite[Lemma 2.29]{tulino2004random}
\begin{eqnarray}
\breve{\matr{\Omega}}_{1,1}\convp \breve{\matr{\Omega}}_{2,2} \convp {\rm G}\left( -\frac{1}{\gamma} \right),
\end{eqnarray}
where 
\begin{eqnarray}
{\rm G}(w)=\sqrt{\frac{(1-\alpha)^2}{4w^2}-\frac{1+\alpha}{2w} +\frac{1}{4}}-\frac{1}{2}-\frac{1-\alpha}{2w}
\end{eqnarray}
is the Stieltjes transform of the eigenvalue distribution of $\matr{H}_{\setminus m}\matr{H}_{\setminus m}^\dagger$ \cite{muller2013applications}. Therefore, the two eigenvalues of $\breve{\matr{\Omega}}$ are obtained as
\begin{eqnarray}
\lambda_{m,1}\convp\lambda_{m,2}\convp
\sqrt{\frac{(1-\alpha)^2\gamma^2}{4}+\frac{(1+\alpha)\gamma}{2} +\frac{1}{4}}-\frac{1}{2}+\frac{(1-\alpha)\gamma}{2},
\end{eqnarray}
for $K,N\rightarrow \infty$.

\end{appendices}

\bibliographystyle{IEEEtran} 
\bibliography{lit}

\begin{thebibliography}{10}
\providecommand{\url}[1]{#1}
\csname url@samestyle\endcsname
\providecommand{\newblock}{\relax}
\providecommand{\bibinfo}[2]{#2}
\providecommand{\BIBentrySTDinterwordspacing}{\spaceskip=0pt\relax}
\providecommand{\BIBentryALTinterwordstretchfactor}{4}
\providecommand{\BIBentryALTinterwordspacing}{\spaceskip=\fontdimen2\font plus
\BIBentryALTinterwordstretchfactor\fontdimen3\font minus
  \fontdimen4\font\relax}
\providecommand{\BIBforeignlanguage}[2]{{%
\expandafter\ifx\csname l@#1\endcsname\relax
\typeout{** WARNING: IEEEtran.bst: No hyphenation pattern has been}%
\typeout{** loaded for the language `#1'. Using the pattern for}%
\typeout{** the default language instead.}%
\else
\language=\csname l@#1\endcsname
\fi
#2}}
\providecommand{\BIBdecl}{\relax}
\BIBdecl

\bibitem{saito2013non}
Y.~Saito, Y.~Kishiyama, A.~Benjebbour, T.~Nakamura, A.~Li, and K.~Higuchi,
  ``Non-orthogonal multiple access ({NOMA}) for cellular future radio access,''
  in \emph{IEEE 77th Vehicular Technology Conference (VTC Spring)}, 2013, pp.
  1--5.

\bibitem{ding2014performance}
Z.~Ding, Z.~Yang, P.~Fan, and H.~V. Poor, ``On the performance of
  non-orthogonal multiple access in {5G} systems with randomly deployed
  users,'' \emph{IEEE Signal Processing Letters}, vol.~21, no.~12, pp.
  1501--1505, 2014.

\bibitem{ding2016impact}
Z.~Ding, P.~Fan, and V.~Poor, ``Impact of user pairing on {5G} non-orthogonal
  multiple access downlink transmissions,'' \emph{IEEE Transactions on
  Vehicular Technology}, vol.~65, no.~8, pp. 6010--6023, Aug. 2016.

\bibitem{ding2017application}
Z.~Ding, Y.~Liu, J.~Choi, Q.~Sun, M.~Elkashlan, I.~Chih-Lin, and H.~V. Poor,
  ``Application of non-orthogonal multiple access in {LTE} and {5G} networks,''
  \emph{IEEE Communications Magazine}, vol.~55, no.~2, pp. 185--191, 2017.

\bibitem{chen2016mathematical}
Z.~Chen, Z.~Ding, X.~Dai, and R.~Zhang, ``A mathematical proof of the
  superiority of {NOMA} compared to conventional {OMA},'' \emph{arXiv preprint
  arXiv:1612.01069}, 2016.

\bibitem{tabassum2016non}
H.~Tabassum, M.~S. Ali, E.~Hossain, M.~Hossain, D.~I. Kim \emph{et~al.},
  ``Non-orthogonal multiple access ({NOMA}) in cellular uplink and downlink:
  Challenges and enabling techniques,'' \emph{arXiv preprint arXiv:1608.05783},
  2016.

\bibitem{timotheou2015fairness}
S.~Timotheou and I.~Krikidis, ``Fairness for non-orthogonal multiple access in
  {5G} systems,'' \emph{IEEE Signal Processing Letters}, vol.~22, no.~10, pp.
  1647--1651, 2015.

\bibitem{wei2016survey}
Z.~Wei, J.~Yuan, D.~W.~K. Ng, M.~Elkashlan, and Z.~Ding, ``A survey of downlink
  non-orthogonal multiple access for {5G} wireless communication networks,''
  \emph{arXiv preprint arXiv:1609.01856}, 2016.

\bibitem{benjebbour2013system}
A.~Benjebbour, A.~Li, Y.~Saito, Y.~Kishiyama, A.~Harada, and T.~Nakamura,
  ``System-level performance of downlink {NOMA} for future {LTE}
  enhancements,'' in \emph{Globecom Workshops (GC Wkshps), 2013 IEEE}.\hskip
  1em plus 0.5em minus 0.4em\relax IEEE, 2013, pp. 66--70.

\bibitem{nikopour2013sparse}
H.~Nikopour and H.~Baligh, ``Sparse code multiple access,'' in \emph{Personal
  Indoor and Mobile Radio Communications (PIMRC), 2013 IEEE 24th International
  Symposium on}.\hskip 1em plus 0.5em minus 0.4em\relax IEEE, 2013, pp.
  332--336.

\bibitem{nikopour2014scma}
H.~Nikopour, E.~Yi, A.~Bayesteh, K.~Au, M.~Hawryluck, H.~Baligh, and J.~Ma,
  ``{SCMA} for downlink multiple access of {5G} wireless networks,'' in
  \emph{Global Communications Conference (GLOBECOM), 2014 IEEE}.\hskip 1em plus
  0.5em minus 0.4em\relax IEEE, 2014, pp. 3940--3945.

\bibitem{tse:05}
D.~Tse and P.~Viswanath, \emph{Fundamentals of Wireless Communications}.\hskip
  1em plus 0.5em minus 0.4em\relax New York: Cambridge University Press, 2005.

\bibitem{caire2007hard}
G.~Caire, R.~R. M\"uller, and R.~Knopp, ``Hard fairness versus proportional
  fairness in wireless communications: The single-cell case,'' \emph{IEEE
  Transactions on Information Theory}, vol.~53, no.~4, pp. 1366--1385, 2007.

\bibitem{higuchi2015non}
K.~Higuchi and A.~Benjebbour, ``Non-orthogonal multiple access ({NOMA}) with
  successive interference cancellation for future radio access,'' \emph{IEICE
  Transactions on Communications}, vol.~98, no.~3, pp. 403--414, 2015.

\bibitem{ding2015cooperative}
Z.~Ding, M.~Peng, and H.~V. Poor, ``Cooperative non-orthogonal multiple access
  in {5G} systems,'' \emph{IEEE Communications Letters}, vol.~19, no.~8, pp.
  1462--1465, 2015.

\bibitem{sun2016optimal}
Y.~Sun, D.~W.~K. Ng, Z.~Ding, and R.~Schober, ``Optimal joint power and
  subcarrier allocation for {MC-NOMA} systems,'' \emph{arXiv preprint
  arXiv:1603.08132}, 2016.

\bibitem{he2016fast}
J.~He, Z.~Tang, and Z.~Che, ``Fast and efficient user pairing and power
  allocation algorithm for non-orthogonal multiple access in cellular
  networks,'' \emph{Electronics Letters}, vol.~52, no.~25, pp. 2065--2067,
  2016.

\bibitem{mueller:99}
R.~R. M{\"u}ller, \emph{Power and Bandwidth Efficiency of Multiuser Systems
  with Random Spreading}.\hskip 1em plus 0.5em minus 0.4em\relax Aachen,
  Germany: Shaker--Verlag, 1999.

\bibitem{al2014uplink}
M.~Al-Imari, P.~Xiao, M.~A. Imran, and R.~Tafazolli, ``Uplink non-orthogonal
  multiple access for {5G} wireless networks,'' in \emph{11th International
  Symposium on Wireless Communications Systems (ISWCS)}.\hskip 1em plus 0.5em
  minus 0.4em\relax IEEE, 2014.

\bibitem{al2015receiver}
M.~Al-Imari, P.~Xiao, and M.~A. Imran, ``Receiver and resource allocation
  optimization for uplink {NOMA} in {5G} wireless networks,'' in \emph{Wireless
  Communication Systems (ISWCS), 2015 International Symposium on}.\hskip 1em
  plus 0.5em minus 0.4em\relax IEEE, 2015, pp. 151--155.

\bibitem{sun2015ergodic}
Q.~Sun, S.~Han, I.~Chin-Lin, and Z.~Pan, ``On the ergodic capacity of {MIMO}
  {NOMA} systems,'' \emph{IEEE Wireless Communications Letters}, vol.~4, no.~4,
  pp. 405--408, 2015.

\bibitem{ding2016application}
Z.~Ding, F.~Adachi, and H.~V. Poor, ``The application of {MIMO} to
  non-orthogonal multiple access,'' \emph{IEEE Transactions on Wireless
  Communications}, vol.~15, no.~1, pp. 537--552, 2016.

\bibitem{ding2016generalTWCOM}
Z.~Ding, R.~Schober, and H.~V. Poor, ``A general {MIMO} framework for {NOMA}
  downlink and uplink transmission based on signal alignment,'' \emph{IEEE
  Transactions on Wireless Communications}, vol.~15, pp. 4438--4454, June 2016.

\bibitem{ding2016design}
Z.~Ding and H.~V. Poor, ``Design of massive-{MIMO}-{NOMA} with limited
  feedback,'' \emph{IEEE Signal Processing Letters}, vol.~23, no.~5, pp.
  629--633, 2016.

\bibitem{cover:91}
T.~M. Cover and J.~A. Thomas, \emph{Elements of Information Theory}.\hskip 1em
  plus 0.5em minus 0.4em\relax New York: John Wiley \& Sons, 1991.

\bibitem{kuhn1955hungarian}
H.~W. Kuhn, ``The hungarian method for the assignment problem,'' \emph{Naval
  research logistics quarterly}, vol.~2, no. 1-2, pp. 83--97, 1955.

\bibitem{lawler2001combinatorial}
E.~L. Lawler, \emph{Combinatorial optimization: networks and matroids}.\hskip
  1em plus 0.5em minus 0.4em\relax Courier Corporation, 2001.

\bibitem{makohon2016hungarian}
I.~Makohon, M.~Cetin, D.~T. Nguyen, and M.~Ng, ``Hungarian optimum assignment
  algorithm with {Java} computer animation,'' in \emph{SoutheastCon,
  2016}.\hskip 1em plus 0.5em minus 0.4em\relax IEEE, 2016, pp. 1--7.

\bibitem{bhatia1997matrix}
R.~Bhatia, \emph{Matrix analysis}.\hskip 1em plus 0.5em minus 0.4em\relax
  Springer Science \& Business Media, 2013, vol. 169.

\bibitem{jain1984quantitative}
R.~Jain, D.-M. Chiu, and W.~R. Hawe, \emph{A quantitative measure of fairness
  and discrimination for resource allocation in shared computer system}.\hskip
  1em plus 0.5em minus 0.4em\relax Eastern Research Laboratory, Digital
  Equipment Corporation Hudson, MA, 1984, vol.~38.

\bibitem{david1981order}
H.~A. David and H.~N. Nagaraja, \emph{Order statistics}.\hskip 1em plus 0.5em
  minus 0.4em\relax Wiley Online Library, 1981.

\bibitem{tulino2004random}
A.~M. Tulino and S.~Verd{\'u}, \emph{Random matrix theory and wireless
  communications}.\hskip 1em plus 0.5em minus 0.4em\relax Now Publishers Inc,
  2004, vol.~1.

\bibitem{muller2013applications}
R.~R. M{\"u}ller, G.~Alfano, B.~M. Zaidel, and R.~de~Miguel, ``Applications of
  large random matrices in communications engineering,'' \emph{arXiv preprint
  arXiv:1310.5479}, 2013.

\end{thebibliography}

\end{document}